\documentclass[11pt]{article}
\usepackage[utf8]{inputenc}
\usepackage[margin=2.54cm]{geometry}
\usepackage{graphicx}
\usepackage{amsmath}
\usepackage{amsthm}
\usepackage{amsfonts}
\usepackage{amssymb}
\usepackage{xspace}
\usepackage[usenames, dvipsnames]{xcolor}
\usepackage{bbm}
\usepackage{soul}
\usepackage{tabularx,ragged2e} 
\usepackage{hhline}
\usepackage{float}
\usepackage{nicefrac}
\usepackage{csquotes}
\MakeOuterQuote{"}
\usepackage[margin=1cm]{caption}
\usepackage[colorlinks]{hyperref}
\usepackage{cleveref}
\usepackage[noend]{algpseudocode}
\usepackage{algorithm}

\DeclareMathOperator{\polylog}{polylog}

\DeclareMathOperator{\range}{range}

\newtheorem{theorem}{Theorem}

\newtheorem{lemma}[theorem]{Lemma}

\newtheorem{claim}[theorem]{Claim}

\newcommand{\EE}{\mathbb{E}}

\newcommand{\ZZ}{\mathbb{Z}}
\newcommand{\indic}{\mathbbm{1}}
\newcommand{\cA}{\mathcal{A}}
\newcommand{\cB}{\mathcal{B}}
\newcommand{\cC}{\mathcal{C}}
\newcommand{\cD}{\mathcal{D}}

\newcommand{\cX}{\mathcal{X}}

\newcommand{\tO}{\widetilde{O}}
\newcommand{\tOmega}{\widetilde{\Omega}}

\newcommand{\tmu}{\widetilde{\mu}}

\newcommand{\ceil}[1]{{\left\lceil{#1}\right\rceil}}
\newcommand{\floor}[1]{{\left\lfloor{#1}\right\rfloor}}

\newcommand{\norm}[1]{{\left\Vert{#1}\right\Vert}}

\newcommand{\avoid}{\textsc{avoid}\xspace}
\newcommand{\mif}{\textsc{mif}\xspace}
\newcommand{\init}{\mathrm{init}\xspace}

\title{Streaming algorithms for the missing item finding problem}
\author{Manuel Stoeckl\thanks{Department of Computer Science, Dartmouth College.\newline  This work was supported in part by the National Science Foundation under award 2006589.}}
\date{}

\begin{document}

\maketitle

\begin{abstract}
  Many problems on data streams have been studied at two extremes of difficulty:
  either allowing randomized algorithms, in the static setting (where they should
  err with bounded probability on the worst case stream); or when only 
  deterministic and infallible algorithms are required. Some recent works have considered
  the adversarial setting, in which a randomized streaming algorithm must succeed
  even on data streams provided by an adaptive adversary that can see the 
  intermediate outputs of the algorithm.
  
  In order to better understand the differences between these models, we study
  a streaming task called ``Missing Item Finding''.
  In this problem, for $r < n$, one is given a data stream $a_1,\ldots,a_r$ of
  elements in $[n]$, (possibly with repetitions), and must output some $x \in [n]$ 
  which does not equal any of the $a_i$. We prove that, for $r = n^{\Theta(1)}$ and
  $\delta = 1/\mathrm{poly}(n)$,
  the space required for randomized algorithms that solve this problem in the
  static setting with error $\delta$ is $\Theta(\mathrm{polylog}(n))$;
  for algorithms in the adversarial setting with error $\delta$, $\Theta((1 + r^2 / n)  \mathrm{polylog}(n))$;
  and for deterministic algorithms, $\Theta(r / \mathrm{polylog}(n))$. Because our
  adversarially robust algorithm relies on free access to a string of $O(r \log n)$
  random bits, we investigate a ``random start'' model of streaming algorithms
  where all random bits used are included in the space cost. Here we find a
  conditional lower bound on the space usage, which depends on the space that
  would be needed for a pseudo-deterministic algorithm to solve the problem.
  We also prove an $\Omega(r / \mathrm{polylog}(n))$ lower bound for the space needed
  by a streaming algorithm with $< 1/2^{\mathrm{polylog}(n)}$ error against ``white-box''
  adversaries that can see the internal state of the algorithm, but not predict
  its future random decisions.
\end{abstract}

\section{Introduction}

A streaming algorithm is one which processes a long sequence of input data and
performs a computation related to it. In general, we would like such algorithms
to use as little memory as possible -- preferably far less than the length of the
input -- while producing incorrect output with as low a probability as possible.
For some problems, there is a space-efficient deterministic algorithm, which
works for all possible inputs; but many others require randomized algorithms
which, for any input, have a bounded probability of failure.

In the \emph{adversarial setting}\cite{BenEliezerJWY20}, one considers the case
where a randomized algorithm is processing an input stream that is produced in
real time, and furthermore the algorithm continually produces outputs depending
on the partial stream that it has seen so far. It is possible that the outputs
of the streaming algorithm will affect the future contents of the input stream;
whether by accident or malice, this feedback may yield an input stream for which
the randomized algorithm gives incorrect outputs. Thus, in the adversarial setting,
we require that an algorithm has a bounded probability of failure, even when
the input stream is produced by an adversary that can see all past outputs of
the algorithm.

The extent to which an algorithm is vulnerable to adversaries depends critically
on the use of randomness by the algorithm. If, given a randomized algorithm that 
has nonzero failure probability on any fixed input stream, an adversary somehow
manages to determine all the past and future random choices made by an instance
of the the algorithm, then the adversary can determine a specific continuation
of the input stream on which the instance fails. Algorithms that are robust to
adversaries often prevent the adversary from learning any of their important
random decisions, and ensure that the decisions which are revealed do not affect
the future performance of the algorithm. For example, \cite{BenEliezerJWY20} 
mentions a sketch-switching method in which a robust algorithm maintains multiple
independent copies of a non-robust algorithm; it emits output derived from one
non-robust instance until it reaches the point where an adversary might make
the instance fail, at which point the algorithm switches to another instance,
none of whose random choices have been revealed to the adversary yet.

Recent research has introduced models with requirements stronger than adversarial
robustness. In the white-box streaming model\cite{AjtaiBJSSWZ22}, algorithms must avoid errors
even when the adversary can see the current state of the algorithm (i.e, including
past random decisions), but not future random decisions. In the pseudo-deterministic model\cite{GoldwasserGMW20}, streaming algorithms should with high probability 
always give the same output for a given input; such algorithms are automatically
robust against adversaries, because (assuming the algorithm has not failed)
the outputs of the algorithm reveal nothing about any random decisions made by
the algorithm.

In order to better understand the differences between all these models, we study 
a streaming problem known as Missing Item Finding ($\mif$). This problem is perhaps
the simplest search problem for data streams where the space of possible answers
shrinks as the stream progresses. For parameters $r < n$, given an data stream $a_1,\ldots,a_r$ of length $r$, 
where each element $e_i$ is an integer in the range $[n]$, the goal of the 
$\mif(n,r)$ problem is to identify some integer $x \in [n]$ for which, 
for all $i \in [r]$, $x \ne a_i$.

This problem is of interest because it has significantly different space complexities
for regular randomized streaming algorithms, adversarially robust streaming algorithms,
and deterministic streaming algorithms. Surprisingly, our adversarially robust algorithm
when $r = \sqrt{n}$ needs oracle access to $\tO(\sqrt{n})$ random bits, but only $\tO(\log n)$
random bits of mutable memory. One of the main open problems left by our work is whether this
is necessary. Our white-box model lower bound shows that the algorithm must make at
least \emph{some} random decisions that remain hidden from the adversary, and a 
conditional lower bound shows that, if the pseudo-deterministic space complexity
of $\mif(n,\sqrt{n})$ is $\tOmega(\sqrt{n})$, then the robust algorithm actually
must use $\tOmega(n^{1/4})$ bits of space, including random bits.

\subsection{Our results and contributions}

Our results -- a series of upper and lower bounds for space complexities of $\mif(n,r)$ in various different models are given in Table \ref{table:results}. For a more precise description of the models and of what guarantee exactly $\delta$ is associated with in each case, see \Cref{sec:prelim}.

\begin{table}[h!]
  \newcolumntype{P}[1]{>{\RaggedRight\arraybackslash}p{#1}}
  \renewcommand{\arraystretch}{1.5}
  \begin{tabular}{l | P{6cm} | P{5cm} | P{3cm}}
  Model & Lower bound & Upper bound & Source\\ \hhline{=|=|=|=}
  Classical     & $\Omega(\sqrt{\frac{\log(1/\delta)}{\log(n)}} + \frac{\log(1/\delta)}{(\log n) (1 + \log(n/r))} )$ ${\text{if } \delta \ge 1/n^r}$
                & $\min(r, \frac{\log(1/\delta)}{\log(n/r)})$ 
                & Thm \ref{thm:lb-classical}, Thm  \ref{thm:ub-classical} \\ \hline
  Adv. Robust   & $\Omega(\frac{r^2}{n} + \log(1 - \delta))$ 
                & $O(\min(r, \left(1 + \frac{r^2}{n} + \ln\frac{1}{\delta}\right) \cdot \log r))$
                & Thm \ref{thm:lb-advrobust} , Thm \ref{thm:ub-advrobust} \\ \hline
  Zero error $\star$    & $\Omega(\frac{r^2}{n})$
                & $O(\min(r, (1 + \frac{r^2}{n}) \log r )$ 
                & Thm \ref{thm:lb-zero-error}, Thm \ref{thm:ub-zero-error} \\ \hline
  Deterministic & $\Omega(\sqrt{r} + \frac{r}{1 + \log(n/r)})$ 
                & $O(\sqrt{r \log r} + \frac{r \log r}{\log n})$
                & Thm \ref{thm:lb-deterministic}, Thm \ref{thm:ub-deterministic}, \\ \hline
  White box     & $\Omega(r/ (\log n)^4)$ ${\text{if } \delta \le 1/n^{O(\log n)}}$
                & (see deterministic)
                & Thm \ref{thm:lb-whitebox} \\ \hline
  Random start  & $\Omega(\sqrt{r} / \polylog n)$, assuming \mbox{Pseudo-deterministic} algs require $\Omega(r / \polylog n)$ bits
                & $O((\sqrt{r} + r^2/n) \log n)$
                & Thm \ref{thm:lb-randstart}, Thm \ref{thm:ub-randstart}  \\ \hline
  \end{tabular}
  \caption{\label{table:results} Table summarizing the upper and lower bounds on the space complexity of algorithms for $\mif(n,r)$ in various models. $\delta$ is the worst case error -- see \Cref{sec:prelim} for what this means in the different models. $\star$: Unlike the other models, the complexity bounds for the zero error case are defined using of the \emph{expected} algorithm space usage, not the worst-case space usage.}
\end{table}

We shall highlight some of the more novel results in what follows:
\begin{itemize}
  \item Our adversarially robust algorithm for $\mif(n,r)$ uses its oracle-type access to random bits to keep track of a list $L$ of outputs that it could give. At each point in time, \Cref{alg:hidden-list} outputs the first element of $L$ which is still available.
  An adversary can choose to make the algorithm move to the next list element, but it
  cannot reliably provide an element from $L$ that it has not yet seen. For the algorithm, switching to the next list element is easy -- it just increments a counter -- but keeping track of future intersections between the $L$ and the stream requires that
  it record each intersecting element; fortunately, even with an adversary there will not be too many such intersections.

  \item  
    Our deterministic algorithm for $\mif(n,r)$ uses the (missing-) pigeonhole
    principle multiple times, and stays within a factor $\log r$ of the space lower bound.
    \Cref{alg:iterated-pigeonhole} proceeds in several stages; in each stage, it
    considers a partition of the input space into a number of different parts, and maintains
    a bit vector keeping track of which part contains an element from
    the stream that arrived in the current stage. When there is exactly one part 
    left, the algorithm remembers that part, discards the bit vector, and moves on
    to the next stage and a new partition of the input space. With suitably chosen
    partitions, the intersection of all the remembered parts from the different 
    stages will be nonempty and disjoint from each element of the stream. The
    algorithm then reports an element from this intersection.
  
  \item Our white-box lower bound proof establishes an adversary that samples
    its next batch of inputs using a distribution $\nu$ over $[n]$ which is
    chosen so that the algorithm will also produce outputs distributed
    according to $\nu$. This is done using recursive applications of Brouwer's
    fixed point theorem: for example, at the base level, we can use it because the map from
    the distribution on $[n]$ out of which the remaining input
    elements are sampled, to the distribution of the final algorithm output,
    is a continuous map from the space of distributions on $[n]$ to itself.
    Note that if $\nu$ picks some element with probability $\ge 2/3$, then the algorithm
    will also output that element with probability $\ge 2/3$, leading to a 
    $\ge 1/3$ chance that the algorithm incorrectly emits an output that 
    it received in the stream. We then show that, if a white box algorithm 
    using less space than our lower bound exists, then said algorithm
    will fail with $\ge 1/2^{O((\log n)^2)}$ probability. This follows by an 
    inductive argument which shows that, at any point in the stream, either
    the algorithm will make a mistake with significant probability, or
    there is a large enough chance that the next distribution which the
    adversary picks will be more "concentrated" than before, as measured by
    an $\ell_p$ norm for a value of $p$ slightly larger than 1.
    As distributions cannot be infinitely "concentrated", it follows that
    the algorithm will eventually make a mistake with some low probability.
    
  \item Our conditional lower bound proof for the "random start" model, relies on the observation that at a given point in the stream, either the adversary is able
  to provide an input where it learns a lot about the initial random bits of the
  algorithm, or the algorithm, because it reveals very little about its internal randomness, also must consistently produce the same output at some point,
  in response to the same input. We can use this behavior to construct a pseudo-deterministic algorithm which works on a shorter input stream.

\end{itemize}

The rest of this paper is organized as follows. Related work is described in \Cref{sec:related-work}. Detailed descriptions of the models for streaming algorithms are given in \Cref{sec:prelim}. Sections \ref{sec:classical} through \ref{sec:randstart-and-pd} contain the main results of this paper, organized according to the rows of \Cref{table:results}; they can be read in any order.

\section{Related work}\label{sec:related-work}

The Missing Item Finding problem appears to have been first studied by \cite{Tarui07}. While they primarily consider the problem of finding a duplicate element in a stream of $m > n$ elements chosen from $[n]$, most of their results also apply to $\mif(n,n-1)$. For example, their multi-pass duplicate finding algorithms can easily be translated to multiple pass algorithms to find a missing element. Their main results also hold: they find an deterministic streaming algorithm for $\mif(n,n-1)$ using $O(\log n)$ bits of space must make $\Omega(\log n / \log \log n)$ passes over the stream, and claim that a single-pass deterministic algorithm for $\mif(n,n-1)$ requires at least $2^n - 1$ states.\footnote{As \Cref{alg:trivial} uses exactly $2^{n-1}$ states for $\mif(n,n-1)$, the value $2^{n} - 1$ may be a typo.}

A variation on the Missing Item Finding problem, that forbids repeated elements in the input stream, was briefly studied in the first section of \cite{Muthukrishnan05}. The paper mentions that for any $k \ge 1$, on a stream encoding a subset of $[n]$ of size $n-k$, it is possible to recover the remaining $k$ elements with a sketch of size $O(k \log n)$. The paper \cite{ChakrabartiGS22} also briefly mentions a variant of Missing Item Finding to illustrate an exponential gap between space usage for regular randomized and adversarially robust streaming algorithms. For the problem where the stream can list any strict subset $S$ of $[n]$, and one must recover a single element not in $S$, they observe that there is a randomized algorithm which uses an $L_0$-sampling sketch to solve the problem in $O((\log n)^2))$ space; but any adversarially robust algorithm that succeeds with high probability needs $\Omega(n)$ bits.

If we were to extend the Missing Item Finding problem to turnstile streams, then we would end up with something opposite to the "support-finding" streaming problem. In the support-finding problem, the algorithm is given a turnstile stream of updates to a vector $x \in \ZZ^{[n]}$; on querying the algorithm, it must return any index $i \in [n]$ where $x_i \ne 0$. \cite{KapralovNPWWY17} find that this problem -- and the harder $L_0$ sampling problem, where one must find a uniformly random element of the support of $x$ -- have a space lower bound of $\Omega\left(\min\left(n, \log\frac{1}{\delta} (\log\frac{n}{\log(1/\delta)} )^2 \right)\right)$. This is close to \cite{JowhariST11}'s $L_0$ sampling algorithm which uses $O( \log \frac{1}{\delta} (\log n)^2)$ bits of space.

The paper \cite{MenuhinN22} studies a two player game that is similar to Missing Item Finding. Here there are two players, a "Dealer" and a "Guesser": for each of $n$ turns, the players simultaneously do the following: the Dealer chooses a number from $[n]$ that it has not picked so far, and the Guesser guesses a number in $[n]$. The goal of the Guesser is to maximize expected score, the number of times their number matches the Dealer's choice; the Dealer tries to minimize the score. The paper proves upper and lower bounds on the expected score, for a number of scenarios. Notably, a Guesser that is limited to remember only $m$ bits of information can do much better against a static Dealer
(that chooses a hard ordering of numbers at the start of the game) than against an adaptive Dealer (that may choose the next number depending on the guesses made by the Guesser.) For example, $m = O((\log n)^2)$ suffices for an expected score of $\Omega(\log n)$ against a static Dealer, but there exists an adaptive Dealer which limits any Guesser's expected score to $(1 + o(1)) \ln m + O(\log \log n)$. The objectives of the Guesser and Dealer are similar to those of the algorithm and adversary in Missing Item Finding: the Guesser tries to avoid, if possible, guessing any value that the Dealer has revealed before; while the Dealer tries to ensure the Guesser chooses that the Dealer had already sent before. However, unlike Missing Item Finding, the Dealer-Guesser game requires that numbers dealt never be repeated and that all numbers be used, which makes it much easier to identify a number that will be dealt in the future.

In the Mirror Game of \cite{GargS18}, there are two players, Alice and Bob who alternately declare numbers from the set $[2n]$. The players lose if they declare a number that has been declared before. Since Alice goes first, even if Bob can only remember $O(\log n)$ bits about the history of the game, Bob still has a simple strategy that will not lose. On the other hand, \cite{GargS18} prove that in order for Alice to guarantee a draw against Bob, they require $\Omega(n)$ bits of memory. If a low probability of error is acceptable, \cite{Feige19} provide a randomized strategy for Alice with $O((\log n)^3)$ bits of memory that draws with high probability -- but this requires oracle access to a large number of random bits, or cryptographic assumptions. \cite{Feige19} and \cite{MenuhinN22} ask whether there is a strategy using $O(\polylog n)$ bits of memory and of randomness. (Again, the objective of Alice in this game is quite similar to that of the algorithm in Missing Item Finding -- but numbers are never repeated, and all numbers in $[2n]$ are used by the end of the game.)

The problem of constructing an adversarially resilient Bloom filter is addressed by \cite{NaorY19}. Here one seeks a an "approximate set membership" data structure, which is initialized on a set $S$ of size $n$, and thereafter answers queries of the form "is $x \in S$" with false positive error probability $\epsilon$. An implementation of this structure is adversarially resilient if the false positive probability of the last element in the sequence is still $\le \epsilon$ when the adversary chooses the sets $S$, and adaptively chooses the sequence of $t$ elements to query. In addition to lower and upper bound results conditional on the existence of one-way functions, \cite{NaorY19} find a construction for an adversarially resilient bloom filter using $O(n \log 1/\epsilon + t)$ bits of memory.

There are many papers on the topic of adversarially robust streaming. Among them, we mention \cite{HardtW13}, who prove that linear sketches on turnstile streams are not, in general, robust against adversaries. \cite{BenEliezerY20} find that algorithms based on finding a representative random sample of the elements in a stream may need only slight modification to work with adaptive adversaries; \cite{BenEliezerJWY20} establish general methods to convert streaming algorithms with real valued output that are not robust against adversaries to ones which are, in exchange for an increase in space usage. \cite{HassidimKMMS20} improve on the space tradeoff of this result by using differential privacy. \cite{WoodruffZ22} improve the space/approximation factor tradeoffs for adversarially robust algorithms on tasks like $F_p$ estimation. 

The thread of finding separations between the space needed for classical streaming and for adversarially robust streaming has been pursued by \cite{KaplanMNS21}, who construct a problem whose classical and adversarially robust space complexities are exponentially separated. \cite{ChakrabartiGS22} mention that this also holds for the variant of Missing Item Finding mentioned above, and prove a separation for the adversarially robust space complexity of graph coloring on insertion streams.

Pseudo-deterministic streaming algorithms were first studied by \cite{GoldwasserGMW20}; the paper finds a separation between the classical and pseudo-deterministic memory needed for the task of finding a nonzero entry of a vector given by turnstile updates from a stream, among other problems. While it is not a streaming task, the Find1 query problem -- in which one is given a bit vector $x$ with $\ge \nicefrac{1}{2}$ density of ones, and must find an index $i$ where $x_i = 1$ by querying coordinates -- has been found to require significantly more queries in the pseudo-deterministic case than in the general randomized case \cite{GoldwasserIPS21}.

Streaming algorithms robust against white box adversaries were considered by \cite{AjtaiBJSSWZ22}; they rule out efficient white-box adversarially robust algorithms for tasks like $F_p$ moment estimation, while finding algorithms for heavy-hitters-type problems. They also show how to reduce white-box adversarially robust algorithms to deterministic 2-party communication protocols, where lower bounds may be easier to prove.\footnote{Unfortunately, for Missing Item Finding, the natural 2-party communication game is $\avoid(n,r/2,r/2)$, whose deterministic communication lower bound is almost the same as the randomized lower bound. See \Cref{subsec:lemmas}. In contrast, our deterministic and white box lower bounds both use $O(\log n)$ players/adaptive steps. }

The Missing Item Finding problem has connections to graph streaming problems. Just as the $L_0$-sampling problem has been used by streaming algorithms that find a structure in a graph, behaviors like those of the Missing Item Finding problem appear in algorithms that look for a structure which is not in a graph. Specifically, the graph coloring problem is equivalent to finding a small collection of cliques which cover all vertices but do not include any edge in the graph. \cite{AssadiCK19} proved that general randomized streaming algorithms can $\Delta+1$ color a graph in $\tO(n)$ space, where $n$ is the number of vertices. \cite{ChakrabartiGS22} showed that adversarially robust streaming algorithms in $\tO(n)$ space must use at least $\Delta^2$ colors for a graph of maximum degree $\Delta$; and \cite{AssadiCS22} proved that deterministic streaming algorithms using $\tO(n)$ space must use $\exp(\Delta^{\Omega(1)})$ colors. The papers \cite{ChakrabartiGS22} and \cite{AssadiCS22} are noteworthy in particular because their lower bound proofs use essentially the same arguments as this paper's lower bound proofs for Missing Item Finding. (In fact, our proof of \Cref{thm:lb-advrobust} was inspired by the \cite{ChakrabartiGS22}'s proof, while \Cref{thm:lb-deterministic} was independently developed.) Because of this, we suspect that this paper's white box lower bound will have an analogue for graph coloring.

\section{Preliminaries}\label{sec:prelim}

\paragraph{Notation} In this paper, following standard convention, $[n]$ is the set $\{1,2,\ldots,n\}$, and $\binom{X}{k}$ is shorthand for the set of all subset of $X$ of size $k$. For a finite set $Y$, we let $\triangle Y$ be the set of all probability distributions over $Y$. For $\pi$ a probability distribution over $Y$, we write $\alpha \sim \pi^k$ to mean that $\alpha \in Y^k$ and each coordinate of $\alpha$ is chosen independently at random according to $\pi$. For some $x \in Y$, the distribution $\indic_y$ is value $1$ on $y$ and value $0$ everywhere else; drawing a sample from this distribution will always result in $y$. The $p$-norm of a distribution $\phi$ on $Y$ is written as $\norm{\phi}_p := \left(\sum_{i \in Y} \phi(i)^p\right)^{1/p}$. The notation $[t]^\star$ gives the set of all sequences of elements from $t$, of any length. The empty sequence is written $\epsilon$; a sequence $s \in [t]^\star$ may be written as $(s_1,s_2,\ldots,s_k)$, in which case its length $|s| = k$. 
To concatenate two sequences $a$ and $b$, we write ``$a.b$''.
$\tO(x)$ means $O(x \polylog(x))$, and $\tOmega(x)$ means $\Omega(x / \polylog(x))$,

\paragraph{A simple algorithm}

While in most cases there are more efficient alternatives, this algorithm for $\mif(n,r)$ is particularly simple:

\begin{algorithm}[H]
  \caption{A simple deterministic streaming algorithm for $\mif(n,r)$
    \label{alg:trivial}}

  \begin{algorithmic}[1]
  \Statex \ul{\textbf{Initialization}}:
    \State $x \gets \{0,\ldots,0\}$, a vector in $\{0,1\}^{[r]}$
    
  \Statex
  \Statex \ul{\textbf{Update}($e \in [n]$)}:
    \If{$e \le r$}
      \State $x_e \gets 1$
    \EndIf
    
  \Statex
  \Statex \ul{\textbf{Query}}:
    \If{$\exists j \in [r] : x_j = 0$}
      \State \textbf{output}: $j$
    \Else
      \State \textbf{output}: $r + 1$
    \EndIf
  \end{algorithmic}
\end{algorithm}

\subsection{Models for streaming algorithms}

We now precisely define the models
of streaming computation considered in this paper. We classify the models by the type
of randomness used, the measure of the
cost of the algorithm, the setting in which they are measured, and by any additional constraints.

\paragraph{Randomness} A streaming algorithm for $\mif(n,r)$ has a set $\Sigma$ of possible states; a possibly random initial state $s_{\init} \in \Sigma$, a possibly random transition function $\tau : \Sigma \times [n] \rightarrow \Sigma$, and a possibly random output function $\omega : \Sigma \rightarrow [n]$. The models of this paper will use the following four variations:

\begin{enumerate}
\item \textbf{Random oracle}: The initial state, transition function, and output function may all be random and correlated; i.e, there is a space $\Omega$ and random variable $R$ on that space for which $s_{\init}$ is a function of $R$, and $\tau(s, a) = f(s, a, R)$ for some deterministic function $f : \Sigma \times [n] \times \Omega \rightarrow \Sigma$, and $\omega(s) = g(s, R)$ for some deterministic function $g : \Sigma \times \Omega \rightarrow [n]$. We can view this as the algorithm having access to an oracle for all of its operations, which provides the value of the variable $R$.

\item \textbf{Random tape}: In this case, the initial state, transition function, and output function are all random, but they are uncorrelated; each step $i$ of the algorithm has associated random variables $R_{i,\tau}$ and $R_{i,\omega}$, and all of these variables are independent of each other and of the initial state $s_{\init}$. The transition function of the algorithm is $\tau(s, a) = f(s, a, R_{i,\tau})$ for some $f$, and the output function is $\omega(s) = g(s, R_{i,\omega})$ for some $g$. If the algorithm visits a state twice, the transitions and outputs from that state will be independent. Intuitively, with this type of access to randomness, the algorithm can always sample fresh random bits (i.e, reading forward on a tape full of random bits), but cannot remember them for free.

\item \textbf{Random seed}: Here the initial state $s_{\init}$ may be chosen randomly, but the transition function and output function are deterministic. The algorithm only has access to the randomness it had when it started.

\item \textbf{Deterministic}: The initial state is fixed, and the transition function and output function are deterministic.
\end{enumerate}

These variations are listed in decreasing order of strength; the random oracle model can emulate the random tape model, which is stronger than the random seed model, which is stronger than the deterministic model. Note that the random oracle model, while inconvenient to implement exactly due to the need to store all the random bits used, can be approximated in practice, since a cryptographically secure random number generator can be used to generate all the random bits from a small random seed.\footnote{As the space cost of this seed can be shared between all tasks performed by a computer, we do not account for it in the space cost estimates for this paper.} Of course, if modern CSPRNGs based on functions like AES are broken, or one-way functions are proven not to exist, then the random oracle model may prove unreasonable.

\paragraph{Cost measure} In this paper, the space cost of an algorithm is the worst case value, over all possible streams or adversaries, of either the \emph{maximum} number of bits used by the algorithm, or the \emph{expected} number of bits used by the algorithm. The number of bits required is determined by a prefix-free encoding of the set $\Sigma$ of states as strings in $\{0,1\}^\star$; for most models, we measure the maximum number of bits used, which is $\ceil{\log |\Sigma|}$ for the best encoding.

\paragraph{Setting} The cost of an algorithm, and its probability of an error, are measured against the type of inputs that it is given.

\begin{enumerate}
\item \textbf{Static}: In the static setting, the algorithm should give an incorrect output, on being queried at the end of the stream, with probability $\le \delta$, when it is given any fixed input stream.\footnote{This is a weaker condition than requiring that the entire sequence of intermediate outputs of the algorithm is correct; however, our lower bounds in static and white-box adversarial settings only require this weaker condition.}

\item \textbf{Adversarial}: In the adversarial setting, we consider the algorithm as being part of a two player game between it and an adversary; the algorithm receives a sequence of elements $e_1,\ldots,e_r$ from the adversary, and after each element $e_i$, the algorithm shall produce an output $o_i$ corresponding to the sequence $e_1,\ldots,e_i$. The adversary chooses input $e_i$ based on the transcript $o_0,e_1,o_1,e_2,\ldots,o_{i-1}$ that has been seen so far. The probability that the sequence of outputs produced by the algorithm has an error should be be $\le \delta$, for any adversary.

\item \textbf{White box adversarial}: This is similar to the adversarial setting, except that here the adversary chooses the next input $e_i$ as a function of the current state $s_i$ of the algorithm. Here, the probability that the algorithm should make a mistake when producing an output at the end of the stream should be $\le \delta$.
\end{enumerate}

\paragraph{Extra constraints} A streaming algorithm may be required to be pseudo-deterministic; in other words, for any input stream $\sigma = e_1,\ldots,e_r$, there should be a corresponding output $o_\sigma$ of the algorithm for which the algorithm is considered to have made a mistake if it does not output $o_\sigma$. In other words, the algorithm should (with probability $\ge 1 - \delta)$ behave as if it were deterministic.

A noteworthy constraint which we do not consider in the following set of models, is the requirement that the algorithm detects when its next output is not certain to be correct, and if so, aborts instead of producing the wrong value. Most of the algorithms presented in this paper already have this property -- the one exception, \Cref{alg:classical}, can be patched to do so at the cost of an extra bit of space.

\paragraph{Models} The models of this paper are described by the following table:

\begin{center}
  \par
  \begin{tabular}{l | l | l | l | l}
  Model                 & Setting              & Randomness      &  Cost           & Extra conditions  \\ \hline
  Classical             & Static               & Oracle          &  Maximum space  &       \\
  Robust                & Adversarial          & Oracle          &  Maximum space  &       \\
  Zero error            & Static               & Oracle          &  Expected space &  $\delta = 0$     \\
  Deterministic         & Static               & Deterministic   &  Maximum space  &       \\
  White box robust      & White-box adv.       & Tape            &  Maximum space  &       \\
  Pseudo-deterministic   & Static               & Oracle          &  Maximum space  & Pseudo-deterministic \\
  Random start          & Adversarial          & Seed            &  Maximum space  &       \\
  \end{tabular}
\end{center}

A brief note on the "Zero error" model; this is a special case where the algorithm may be randomized,
but is required to always give correct output for any input stream; unlike the deterministic model,
the cost of the algorithm is the \emph{expected} number of bits of space used by the algorithm. We include
this model because, in many cases, a computer may run many independent instances of a streaming
algorithm, and it is often more important that the instances do not fail than that they hold to 
strict space limits. In this scenario, as long as the expected space used by each algorithm is limited,
and the worst case space usage is not too extreme, by the Chernoff bound it is unlikely that the total
space used by all the instances exceeds the expected space by a significant amount. Unlike the case
for time complexity, where a Las-Vegas algorithm can be obtained by repeating a 
Monte-Carlo algorithm until the solution is verifiably correct, there is no simple
way to construct a single-pass, zero-error streaming algorithm from one with nonzero error.

We use the following notation for the space complexities of these models. The $\delta$-error space complexity of the classical model for a task $T$ is $S_\delta(T)$; for the robust model, $S^{AR}_\delta(T)$, for the zero error model, $S_0(T)$; for the deterministic model, $S^{det}(T)$; for the
white box robust model, $S^{WB}_\delta(T)$; for the pseudo-deterministic model, $S^{PD}_\delta(T)$, and the random start model, $S^{RS}_\delta(T)$. The following relationships follow from the definitions of the models:
\begin{align*}
              &                     & S^{AR}_\delta(T) &\le S^{RS}_\delta(T) & S^{RS}_\delta(T) &\le S^{det}(T) \\
  S_\delta(T) &\le S^{AR}_\delta(T) & S^{AR}_\delta(T) &\le S^{WB}_\delta(T) & S^{WB}_\delta(T) &\le S^{det}(T) \\
              &                     & S^{AR}_\delta(T) &\le S^{PD}_\delta(T) & S^{PD}_\delta(T) &\le S^{det}(T)\\
              &                     &                  &                     & S_0(T)           &\le S^{det}(T)
   &  \\
\end{align*}

For problems in communication complexity, we write $R^{\rightarrow}_\delta(T)$ for the one-way randomized $\delta$-error communication complexity of task $T$, and $D^{\rightarrow}(T)$ for the deterministic communication complexity.

\subsection{Lemmas}\label{subsec:lemmas}

\paragraph{The $\avoid(t,a,b)$ communication task} This one-way communication game was introduced by \cite{ChakrabartiGS22}. In it, Alice is given $S \subseteq [t]$ with $|S| = a$, and sends a message to Bob, who must produce $T \subseteq [t]$ with $|T| = b$ where $T$ is disjoint from $S$. 

\begin{lemma}(From \cite{ChakrabartiGS22}, Lemma 6)\label{lem:avoid-lb}
  The public-coin $\delta$ error one-way communication complexity of $\avoid(t,a,b)$ is at least $\log(1 - \delta) + \log(\binom{t}{a}/\binom{t - b}{a})$. Because
  \begin{align*}
    \binom{t}{a}/\binom{t - b}{a} = \frac{t! (t-a-b)!}{(t-a)! (t-b)!} \ge 2^{\frac{a b}{t \ln 2}}
  \end{align*}
  we have the weaker but more convenient lower bound $R^{\rightarrow}_\delta(\avoid(t,a,b)) \ge \frac{a b}{t \ln 2} + \log(1-\delta)$
\end{lemma}

The above lower bound is mainly useful when $a b \in [t,t^2]$. For smaller inputs:

\begin{lemma}\label{lem:lb-avoid-low-range}
  The public-coin $\delta$-error one-way communication complexity of $\avoid(t,a,b)$ satisfies 
  \begin{align*}
    R^{\rightarrow}_\delta(\avoid(t,a,b)) \ge \min\left(\log(a+1), \log \frac{\ln(1/\delta)}{\ln(et/a)} \right) \,.
  \end{align*}
  For the deterministic case, we have $D^\rightarrow(\avoid(t,a,b)) \ge \log(a + 1)$.
\end{lemma}

\begin{proof}
  Say we have a public coin one-way randomized protocol $\Pi$ for $\avoid(t,a,b)$ with error $\delta$; by the averaging argument, there exists a fixing of the randomness of the protocol,
  which is correct on $\ge 1 - \delta$ of the sets in $\binom{[t]}{a}$. Let $\Psi$ be this deterministic protocol, and let $\hat{m}$ be the number of distinct messages sent by $\Psi$. Each message $i \in [\hat{m}]$ corresponds to some set $B_i$ that Bob outputs on receiving the message. Let $E := \{e_1,\ldots,e_m\}$ be a hitting set for $\{B_i\}_{i\in[\hat{m}]}$ of size $m \le \hat{m}$; i.e, for all $B_i$, there is some $e_j \in B_i$. Let $\cC \subseteq \binom{[t]}{a}$ be the set of inputs for which $\Psi$ is correct;
  we note that no inputs in $\cC$ can contain all of $E$, because if $A \supseteq E$, then every $B_i$ intersects $A$, making the protocol fail. Assuming $m \le a$, we have:
  \begin{align*}
    \delta &\ge 1 - |\cC|/\binom{t}{a} \ge |\{A \in \binom{t}{a} : A \supseteq E \} | /\binom{t}{a} \\
        &= \binom{t - m}{a - m} /\binom{t}{a} = \frac{a \cdot (a-1) \cdots (a-m+1)}{t \cdot (t-1) \cdots (t-m+1)} \ge \left(\frac{a / e}{t}\right)^m \,,
  \end{align*}
  where the last step is derived from the well known inequality $a! \ge (a/e)^a$.
  Rearranging gives $m \ge \ln(1/\delta) / \ln(et/a))$. In the case where $m > a$, this argument does not work, because then $\binom{t - m}{a - m} = 0$. Combining the two cases gives: $\hat{m} \ge m \ge \min(a+1, \ln(1/\delta) / \ln(et/a)))$. Thus $R^\rightarrow_\delta(\avoid(t,a,b)) \ge \log(\min(a+1, \ln(1/\delta) / \ln(et/a)))$.
  

  For general deterministic protocols, we reuse the analysis of randomized protocols with $\delta = 0$, concluding that $D^\rightarrow(\avoid(t,a,b)) \ge \log(a+1)$.
\end{proof}

The following lemma is a simple variation of Chernoff's and Azuma's inequalities; for completeness, we present a proof in \Cref{sec:appendix}.

\begin{lemma}[Modified Azuma's inequality] \label{lem:azumanoff}
  Let $X_1,\ldots,X_n$ be $\{0,1\}$ 
  random variables, with $\EE[X_i \mid X_1 = x_1,\ldots,X_{i-1} = x_{i-1}] \le p$
  for all $i$ and all $x_1,\ldots,x_n \in \{0,1\}^n$. Then
  \begin{align*}
    \Pr\left[\sum_{i=1}^n{X_i} \ge n p (1 + \delta)\right] \le \left(\frac{e^\delta}{(1+\delta)^{1+\delta}}\right)^{n p} \le e^{ -\frac{\delta^2}{2+\delta} n p } \,.
  \end{align*}
\end{lemma}

A number of versions of Brouwer's fixed point theorem have been proven; in this paper, we will use the following, which is equivalent to Corollary 2.15 of \cite{Hatcher-book}.

\begin{lemma}[Brouwer's fixed point theorem]\label{lem:brouwer-fixed-point}
  Every continuous map from a space homeomorphic to an $n$ dimensional-ball to itself has a fixed point.
\end{lemma}

\section{Classical model}\label{sec:classical}

\begin{theorem}\label{thm:lb-classical}
  For any $\delta \le 1/(2n)$, the space complexity for an algorithm solving $\mif(n,r)$ with error $\le \delta$ is
  $S_\delta(\mif(n,r)) \ge S^{\textsc{det}} (\mif(\ceil{\frac{n}{t}}, \floor{\frac{r}{t}} ))$, for $t = \ceil{ \frac{r \log n}{\log \frac{1}{2 \delta}} }$. If we apply the upcoming lower bound from Theorem \ref{thm:lb-deterministic} on the deterministic space complexity of $\mif$, we get:
  \begin{align*}
    S_\delta(\mif(n,r)) &\ge \Omega\left( \sqrt{\min\left(r, \frac{\log(1/\delta)}{\log n}\right)} + \min\left(r, \frac{\log(1/\delta)}{\log n}\right) \frac{1}{1 + \log(n/r)} \right)
  \end{align*}
\end{theorem}


\begin{proof}
  Let $t$ be an integer satisfying $\ceil{\frac{n}{t}}^\floor{\frac{r}{t}} < \frac{1}{\delta}$; setting $t = \ceil{ \frac{r \log n}{\log \frac{1}{2\delta} } }$ suffices, because
  \begin{align*}
    \log \left(\ceil{\frac{n}{t}}^\floor{\frac{r}{t}} \right) \le \floor{\frac{r}{t}} \log \ceil{\frac{n}{t}} \le \frac{r}{t} \log n \le \frac{r}{\ceil{r \log n / \log(1/2\delta) } } \log n \le \frac{\log(1/2\delta)}{\log n} \log n < \log \frac{1}{\delta} \,.
  \end{align*}
  Note also that because $\delta \le 1/(2n)$, $t \le r$, and $\floor{\frac{r}{t}} \ge 1$.

  Given a randomized algorithm $\Pi$ that solves $\mif(n,r)$ with error $\le \delta$ on any input stream, we will show how to construct a randomized algorithm $\Psi$ which solves $\mif(\ceil{n/t},\floor{r/t})$ with the same error probability. As there are only $\ceil{n/t}^\floor{r/t}$ possible input streams for the $\mif(\ceil{n/t},\floor{r/t})$ task, the probability (over randomness used by $\Psi$) of the event $E$ than an instance $A$ of $\Psi$ succeeds on \emph{any} of the streams in $[\ceil{n/t}]^\floor{r/t}$ is $\ge 1 - \delta \ceil{\frac{n}{t}}^\floor{\frac{r}{t}} > 0$. Therefore, by fixing the random bits of $\Psi$ to some value for which the event $E$ occurs, we obtain a deterministic protocol $\Phi$ for $\mif(\ceil{n/t},\floor{r/t})$.

  We now explain the construction of $\Psi$ given $\Pi$. Let $f : [n] \mapsto [\ceil{n/t}]$ be the function given by $f(x) = \floor{x / t}$. For any $y \in [\ceil{n/t}]$, we have that $f^{-1}(y)$ is a nonempty set of size $\le t$. The protocol $\Psi$ starts by initializing an instance $A$ of $\Pi$, and sending it $r - t \floor{r/t}$ arbitrary stream elements.

  When $\Psi$ receives an element $e \in [\ceil{n/t}]$, it sends a sequence of $t$ elements of $[n]$ to $A$, namely, the elements of $f^{-1}(e)$, in arbitrary order, repeating elements if $|f^{-1}(e)| < t$. To output an element, $\Psi$ queries $A$ to obtain $i \in [n]$, and reports $f(i)$. Assuming $A$ did not fail, $f(i)$ is guaranteed to be a correct answer. If we assume for sake of contradiction that $f(i) = e$ for some element $e$ sent to $\Psi$ earlier, then $A$ must have been sent all elements in $f^{-1}(e)$ -- which implies that $i \in f^{-1}(e)$ and that $A$ gave an incorrect output, contradicting the assumption that $f(i) = e$.
  Thus, we have proven that $\Psi$ fails with no greater probability than $\Pi$, which is all that is needed to complete this part of the proof.

  Having shown that $S_\delta(\mif(n,r)) \ge S^{\textsc{det}} (\mif(\ceil{\frac{n}{t}}, \floor{\frac{r}{t}} ))$, we now substitute in the lower bound from Theorem \ref{thm:lb-deterministic}.
  \begin{align*}
   S_\delta(\mif(n,r)) &\ge  S^{\textsc{det}}\left(\mif\left(\ceil{\frac{n}{t}},\floor{\frac{r}{t}}\right)\right) 
      \ge \Omega\left(\max\left(\sqrt{\floor{\frac{r}{t}}}, \frac{\floor{r/t}}{1 + \log(\ceil{n/t}/\floor{r/t})} \right)\right)
  \end{align*}
  Because $\floor{r/t} = \Theta\left(\min\left(r, \frac{\log(1/\delta)}{\log n}\right)\right)$, and $\ceil{n/t}/\floor{r/t} = \Theta(n/r)$, and $\Omega(\max(a,b))=\Omega(a+b)$, this simplifies to:
  \begin{align*}
    S_\delta(\mif(n,r)) = \Omega\left( \sqrt{\min\left(r, \frac{\log(1/\delta)}{\log n}\right)} + \min\left(r, \frac{\log(1/\delta)}{\log n}\right) \frac{1}{1 + \log(n/r)} \right)
  \end{align*}
\end{proof}

\subsection{Upper bound: a sampling algorithm}

\begin{figure}[ht]
  \centering
  \includegraphics[width=10cm]{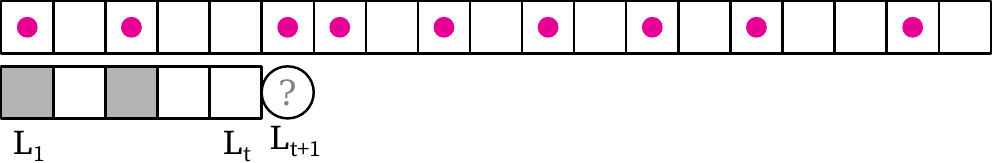}
  \caption{This diagram shows the behavior of \Cref{alg:classical} on an example input. The top row of squares corresponds to the set $[n]$, ordered so that the leftmost squares corresponds to the elements $L_1$, $L_2$, $\ldots$, $L_{t+1}$ from \Cref{alg:classical}. In the top row, cells contain a pink dot if the corresponding element has already been seen in the stream. In the bottom row, each of the cells is shaded dark if the corresponding entry in the vector $x$ is equal to $1$ -- except for $L_{t+1}$, whose state \Cref{alg:classical} does not track.}
\end{figure}

\begin{theorem}\label{thm:ub-classical}
 \Cref{alg:classical} solves $\mif(n,r)$ with error $\le \delta$ on any fixed input stream, and uses $s \le \min(r, \frac{\log(1/\delta)}{\log(n/r)})$ bits of space. (This assumes oracle access to $O((s+1) \log n)$ random bits.)
\end{theorem}

\begin{algorithm}[!ht]
  \caption{A streaming algorithm for $\mif(n,r)$ with error rate $\le \delta$ on any input stream 
    \label{alg:classical}}

  \begin{algorithmic}[1]
  \Statex Let $t = \min(r, \floor{\log(1/\delta)/\log(n/r)})$
  \Statex 
  
  \Statex \ul{\textbf{Initialization}}:
    \State Let $L = \{L_1,\ldots,L_{t+1}\}$ be a fixed sequence of elements in $[n]^{t+1}$ without repetitions, chosen uniformly at random. \textcolor{darkgray}{(This can be stored explicitly using $O((t+1) \log n)$ bits, or computed on demand as a function of $O((t+1) \log n)$ oracle random bits.)} 
    \State $x \gets \{0,\ldots,0\}$, a vector in $\{0,1\}^t$
    
  \Statex
  \Statex \ul{\textbf{Update}($e \in [n]$)}:
    \If{$\exists j \in [t] : L_j = e$}
      \State $x_j \gets 1$
    \EndIf
    
  \Statex
  \Statex \ul{\textbf{Query}}:
    \If{$\exists j \in [t] : x_j = 0$}
      \State \textbf{output}: $L_j$\label{line:ret-tracked}
    \Else
      \State \textbf{output}: $L_{t+1}$\label{line:ret-untracked}
    \EndIf
  \end{algorithmic}
  
\end{algorithm}

\begin{proof}
  First, we observe that \Cref{alg:classical} gives an incorrect output only when the input stream $\sigma = (e_1,\ldots,e_r)$ contains every element of $L$. Otherwise,
  either the first $t$ elements of $L$ are in $\sigma$, and $L_{t+1}$ isn't -- in which case Line \ref{line:ret-untracked} returns $L_{t+1}$ -- or there is some $j \in [t]$ where $L_j$ has not been seen in the stream so far, in which case Line \ref{line:ret-tracked} correctly returns $L_j$. Given a fixed input stream $\sigma \in [n]^r$, the probability that \Cref{alg:classical} fails is:
  \begin{align*}
    \Pr[L \subseteq \sigma] = \binom{|\sigma|}{t+1} \big/ \binom{n}{t+1} \le \binom{r}{t+1} \big/ \binom{n}{t+1} = \frac{r(r-1)\cdots(r-t)}{n(n-1)\cdots(n-t)} \le \left(\frac{r}{n}\right)^{t+1}
  \end{align*}
  Thus $\Pr[L \subseteq \sigma]$ is $\le \delta$ when $t = \floor{\log(1/\delta)/\log(n/r)}$, and is equal to $0$ when $t = r$, because no set of size $r$ can contain a set of size $r+1$.
\end{proof}

\section{Adversarially robust model}\label{sec:advrobust}

\begin{theorem}\label{thm:lb-advrobust}Any algorithm which solves $\mif(n,r)$ against adaptive adversaries with total error $\delta$ requires $\ge \log(\binom{n}{\ceil{r/2}} / \binom{n - \ceil{r/2}}{\floor{r/2} + 1}) + \log(1-\delta)$ bits of space; or less precisely, $\Omega(r^2/n + \log(1-\delta))$.
\end{theorem}

\begin{proof}
  We prove this by reducing the communication task $\avoid(n, \ceil{r/2}, \floor{r/2} + 1)$ (see Section \ref{sec:prelim}) to $\mif(n,r)$.
  
  Say Alice is given the set $A \subseteq [n]$ of size $\ceil{r/2}$. They instantiate an instance $\cX$ of the given algorithm for $\mif(n,r)$, and runs it on the partial stream of length $\ceil{r/2}$ containing the elements of $A$ in some arbitrary order. Alice then sends the state of $\cX$ to Bob; since this is a public coin protocol, all randomness can be shared for free. Bob then runs the following adversary against $\cX$; it queries $\cX$ for an element $b_0$, and then provides that element back to $\cX$,
  repeating this process $\floor{r/2} + 1$ times to recover elements $b_0,b_1,\ldots,b_{\floor{r/2} }$. The instance will fail to give correct answers to this adversary with total probability $\le \delta$. If it succeeds, then by the definition of the Missing Item Finding problem, $b_0 \notin A$, $b_1 \notin \{b_1\} \cup A$, and so on;
  thus Bob can report $B := \{b_0,\ldots,b_{\floor{r/2}+1}\}$ as a set of $\floor{r/2}+1$ elements which are disjoint from $A$.
  
  This \avoid protocol implementation uses the same number of bits of communication as $\cX$ does of space. By Lemma \ref{lem:avoid-lb}, it follows $\cX$ needs:
  \begin{align*}
    &\ge \log\left(\binom{n}{\ceil{r/2}} / \binom{n - \floor{r/2} - 1}{\ceil{r/2}}\right) + \log(1-\delta) \\
    &\ge \frac{\ceil{r/2}(\floor{r/2}+1)}{n \ln 2} + \log(1 - \delta) \ge \frac{r^2}{4 n \ln 2} + \log(1-\delta) \,,
  \end{align*}
  bits of space.
\end{proof}

\subsection{Upper bound: the hidden list algorithm}

\begin{figure}[ht]
  \centering
  \includegraphics[width=10cm]{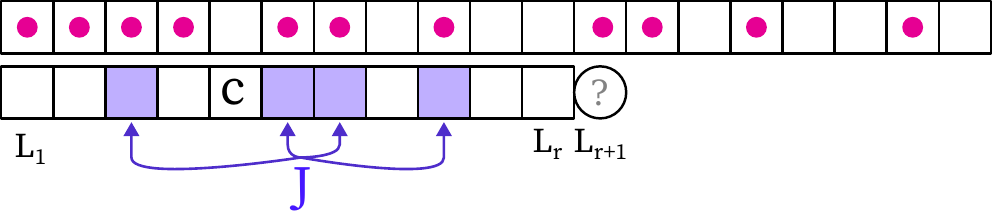}
  \caption{This diagram shows the behavior of \Cref{alg:hidden-list} on an example input. The top row of squares corresponds to the set $[n]$, ordered so that the leftmost squares corresponds to the elements $L_1$, $L_2$, $\ldots$, $L_{r+1}$ from \Cref{alg:hidden-list}. In the top row, cells contain a pink dot if the corresponding element has already been seen in the stream. In the bottom row, the letter C indicates the cell corresponding to $L_c$. Cells that are shaded dark blue indicate the values contained in $J$. The third cell from the left is included in $J$ because, at the time the element $L_3$ was added by the adversary, $c$ was less than or equal to $2$.}
\end{figure}

\begin{theorem}\label{thm:ub-advrobust}
 Algorithm \ref{alg:hidden-list} solves $\mif(n,r)$ against adaptive adversaries, with error $\delta$, and can be implemented using $O(\min(r, \left(1 + \frac{r^2}{n} + \ln\frac{1}{\delta}\right) \cdot \log r))$ bits of space. (It assumes oracle access to $(r+1) \log n$ random bits.)
\end{theorem}

\begin{algorithm}[!ht]
  \caption{An adversarially robust algorithm for $\mif(n,r)$ with error $\le \delta$
    \label{alg:hidden-list}}

  \begin{algorithmic}[1]
  \Statex Let $t = \min(r, \ceil{3 \frac{r^2}{n} + \ln\frac{1}{\delta}})$
  \Statex 
  
  \Statex \ul{\textbf{Initialization}}:
    \State Let $L = \{L_1,\ldots,L_{r+1}\}$ be a fixed sequence of elements in $[n]^{r+1}$ without repetitions, chosen uniformly at random. (Assuming oracle access to $O(r \log n)$  random bits, the value of $L$ can be computed on demand, instead of stored.)
    \State $c \leftarrow 1$, an integer in the range $\{1,\ldots,r+1\}$
    \State $J \leftarrow \emptyset$, a subset of $\{L_1,\ldots,L_r\}$ of size $\le t$
    
  \Statex
  \Statex \ul{\textbf{Update}($e \in [n]$)}:
    \While{$e = L_c$ or $L_c \in J$}
      \State $c \gets c + 1$
    \EndWhile
    \If{$e \in \{L_{c+1},\ldots,L_r\}$}\label{line:hl-jext}
      \State $J \gets J \cup \{e\}$
    \EndIf
    \If{$|J| > t$}
      \State \textbf{abort}
    \EndIf
  
  \Statex
  \Statex \ul{\textbf{Query}}:
    \State \textbf{output}: $L_c$
  
  \end{algorithmic}
    
\end{algorithm}

\begin{proof}
  First, we observe that the only way that \Cref{alg:hidden-list} can fail is if it aborts. At any point in the stream, the set $J$ includes the intersection of the earlier elements from the stream, with the list $\{L_{c+1},\ldots,L_{r}\}$ of possible future outputs. The while loop ensures that the element $L_c$ emitted will neither be equal to the current element nor collide with any past stream elements (those in $J$). It is not possible for $c$ to go out of bounds, because each element in the stream can lead to an increase in $c$ of at most one; either immediately when the element arrives, if $e = L_c$; or delayed slightly, if $e \in \{L_{c+1},\ldots,L_{r}\}$. Since the stream contains $r$ elements, $c$ will increase by at most $r$, to a value of $r + 1$. Note that if $c$ has reached the value $r + 1$, then the entire stream was a permutation of $\{L_1,\ldots,L_r\}$, making $L_{r+1}$ is a safe output.
  
  This algorithm needs $\log(r + 1)$ bits to store $c$, but the main space usage is in storing $J$. We will show that $|J| \le t$ with probability $\ge 1 - \delta$, in which case $J$ can be stored as either a bit vector of length $r$, or a list of $\le t$ indices in $[r]$, using $O(\min(r, t \log r))$ bits of space.
  
  
  We observe that after $i - 1$ elements have been received (and up to $i$ distinct elements emitted), the probability that the $i$th element chosen by the adversary will be newly stored in $J$ will be $\le 2 \frac{r}{n}$, no matter what the earlier elements were or what the adversary picks. If $r \ge n/2$, this is immediate. Otherwise, write $E_{i-1}$ for the set containing the first $i-1$ elements of the stream, $e_i$ for the $i$th element, and let $c_i$ be the value of the variable $c$ as of Line \ref{line:hl-jext}. Let $X_i$ denote the indicator random variable for the event that $e_i$ was not in $J$ before, but has been added now.
  
  Because the adversary has only been given outputs deriving from $L_{\le c_i} := ( L_1,\ldots,L_{c_i} )$,
  if we condition on the random variable $L_{\le c_i}$, then the element $e_i$ and set $E_{i-1}$ are independent of $L_{> c_i} := \{L_{c_i+1},\ldots,L_{r} \}$. Given $E_{i-1}$, the values $X_1,\ldots,X_{i-1}$ determine whether or not each element of $E_{i-1}$ is in $L_{> c_i}$. Then, conditioning on $L_{\le c_i}, e_i, E_{i-1}$, and $X_1,\ldots,X_{i-1}$, we have that $L_{> c_i} \setminus E_{i-1}$ is a set of size $r - c_i - |L_{> c_i} \cap E_{i-1}|$ chosen uniformly at random from $[n] \setminus L_{\le c_i} \setminus E_{i-1}$. Thus, if $e_i \notin L_{\le c_i} \cup E_{i-1}$, the probability that $X_i = 1$ is precisely the probability that $e_i$ is contained in $L_{> c_i} \setminus E_{i-1}$, so:
  \begin{align*}
    \Pr\left[X_i = 1 \mid (X_j)_{j=1}^{i-1}, e_i, E_{i-1}, L_{\le c_i}, \{e_i \notin L_{\le c_i} \cup E_{i-1}\}\right] &= \frac{r - c_i - |L_{> c_i} \cap E_{i-1}|}{n - c_i - |E_{i-1} \setminus L_{\le c_i}|} \\
      &\le \frac{r - c_i}{n - c_i - r} \le \frac{r}{n - r} \le \frac{2 r}{n} \,.
  \end{align*}
  On the other hand, the event $e_i \in L_{\le c_i} \cup E_{i-1}$, implies $X_i = 0$ always. Together, these imply $\Pr[X_i = 1 \mid (X_j)_{j=1}^{i-1}] \le 2r / n$.
  
  Then applying the (modified, see Lemma \ref{lem:azumanoff}) Azuma's inequality bound, we find that with $z := \max\{1, \frac{3 n}{2 r^2} \ln{\frac{1}{\delta}}\}$:
  \begin{align*}
    \Pr[\sum_{i\in[r]} X_i &\le \frac{2r^2}{n} (1 + z)] \le e^{-\frac{z}{2 + z} z \frac{2 r^2}{n}} \\
      &\le e^{- z \frac{2 r^2}{3 n}} && \text{since $z \ge 1$} \\
      &\le e^{- \ln{\frac{1}{\delta}}} = \delta \,. && \text{since $z \ge \frac{3 n}{2 r^2} \ln{\frac{1}{\delta}}$}
  \end{align*}
  This implies that the probability that $|J|$ exceeds $2 r^2 / n + 3 \ln(1/\delta)$ will be $\le \delta$. Consequently, our bound for the total space usage of the algorithm is:
  \begin{align*}
    O(\log r) &+ O(\min(r, \left(\frac{r^2}{n} + \ln\frac{1}{\delta}\right) \log r)) \\
      &= O(\min(r, \left(1 + \frac{r^2}{n} + \ln\frac{1}{\delta}\right) \cdot \log r))
  \end{align*}
\end{proof}

While it is possible to reduce the space usage of \Cref{alg:hidden-list} by removing all elements from the set $J$ that are less or equal than $c$, this only changes the constant factor.

\section{Zero error model}\label{sec:zero-error}

\begin{theorem}\label{thm:lb-zero-error} All algorithms solving $\mif(n,r)$ with zero error on any stream require $\Omega(r^2 / n)$ bits of space, in expectation over the randomness of the algorithm.
\end{theorem}

\begin{proof}
First, we prove that if there is a zero-error algorithm $\Phi$ for $\mif(n,r)$ using exactly $s$ bits, in expectation, then there is a communication protocol for $\avoid(n, \ceil{r/2}, \floor{r/2}+1)$ using prefix-encoded messages with an expected length of $s$ bits. The construction is the same as for Theorem \ref{thm:lb-advrobust}. Alice, on being given a set $A \subseteq [n]$ of size $\ceil{r/2}$, initializes an instance $X$ of $\Phi$, and runs it on an input stream $\alpha$ of length $\ceil{r/2}$ containing each element of $A$ in some arbitrary order. Any random bits used by $X$ are shared publicly with Bob. They send the encoding of $X$'s state to Bob, who queries $X$ to find an element $b_0 \notin \alpha$, updates $X$ with $b_0$, queries it to find $b_1 \notin \alpha \cup \{b_0\}$, and so on until Bob has recovered $B = \{b_0,\ldots,b_{\floor{r/2}}\}$. Because the algorithm is guaranteed to never fail on any input stream, it must in particular succeed on Bob's adaptively chosen continuation of $\alpha$. This ensures that $B \cap A = \emptyset$ holds with probability $1$.

Next, we prove that any zero error randomized communication protocol $\Pi$ for $\avoid(t, a, b)$
requires $\ge a b / (t \ln 2)$ bits \emph{in expectation}. Following the argument from Lemma 6 of \cite{ChakrabartiGS22}, we observe that there must exist a fixing of the public randomness of $\Pi$ for which the expected number of bits used when inputs $A$ are drawn uniformly at random from $\binom{[t]}{a}$, is at least as large as when $\Pi$ is run unmodified. Let $\Upsilon$ be the deterministic protocol with this property, and let $M$ be the set of all messages sent by $\Upsilon$. 
Each message $m \in M$ has a length $|m|$, probability (over the random choice of $A$) $p_m$ of being sent, and makes Bob output the set $B_m$. For all $m \in M$, we have:
\begin{align*}
  p_m = \Pr[\text{$m$ is sent}] \le  \Pr[\text{$B_m$ is a correct output}] = \Pr[A \cap B_m = \emptyset] \le \binom{t - a}{b} / \binom{t}{a} \le 2^{-\frac{a b}{t \ln 2}}.
\end{align*}
Let $\Upsilon(A) \in M$ be the message sent by $\Upsilon$ for a given value of $A$. Then the entropy
\begin{align*}
  H(\Upsilon(A)) = \sum_{m \in M} p_m \log\frac{1}{p_m} = \EE_{A \in \binom{[t]}{a}} \log\frac{1}{p_{\Upsilon(A)}} \ge \EE_{A \in \binom{[t]}{a}} \frac{a b}{t \ln 2} = \frac{a b}{t \ln 2} \,.
\end{align*}
By the source coding theorem,
\begin{align*}
  \EE_{A \in \binom{[t]}{a}} \EE|\Upsilon(A)| \ge H(\Upsilon(A)) \ge \frac{a b}{t \ln 2} \,.
\end{align*}

Applying the above lower bound to the task $\avoid(n, \ceil{r/2}, \floor{r/2 + 1})$, we conclude that $\Phi$ requires $\ge \frac{r^2}{4 n \ln 2}$ bits of space in expectation.

\end{proof}

\begin{theorem}\label{thm:ub-zero-error}
There is an algorithm solving $\mif(n,r)$ with zero error against adaptive adversaries, which uses $O((1 + r^2 /n) \log r)$ bits of space, in expectation over the randomness of the algorithm.
\end{theorem}

\begin{proof}
  We use a slight variation of Algorithm \ref{alg:hidden-list}, in which internal parameter $t$ is instead set to $r$.
  This ensures that the algorithm will never abort; the proof of Theorem \ref{thm:ub-advrobust} 
  has established that Algorithm \ref{alg:hidden-list} will then always give a correct output for the $\mif(n,r)$ task.

  The counter $c$ can be encoded in binary using at most $\ceil{\log(r+1)}$ bits. We encode the set $J$
  by concatenating the binary value of $|J|$, followed by the binary values of the indices $i_1,\ldots,i_{|J|}$ in $[r]$
  for which $L_{i_k}$ is equal to the $k$th smallest element of $J$. (As both the encodings of $c$ and $J$ are prefix codes, so too is the encoding of the algorithm's state formed by concatenating them.) The total space $S$ used by the algorithm (excluding random bits) is then:
  \begin{align*}
    S = \ceil{\log(r+1)} + \ceil{\log{r}} (1 + |J|) \,.
  \end{align*}
  As in the proof of  Theorem \ref{thm:ub-advrobust}, let $X_i$ be the indicator random variable for the event that the $i$th element that the adversary chooses for the stream is stored in $J$; we showed that for all $i \in [r]$, $\Pr[X_i = 1 \mid X_{i-1},\ldots,X_1] \le \frac{r-1}{n}$, which implies $\EE[X_i] \le \frac{r-1}{n}$. By linearity of expectation, 
  \begin{align*}
    \EE S &= \ceil{\log(r+1)} + \ceil{\log{r}} \left(1 + \EE \sum_{i \in [r]} X_i \right) \\
          &\le \ceil{\log(r+1)} + \ceil{\log{r}} \left(1 + \frac{r (r-1)}{n}\right) = O((1 + \frac{r^2}{n}) \log r) \,.
  \end{align*}
\end{proof}

\section{Deterministic model}\label{sec:deterministic}

\subsection{Lower bound: an embedded instance of Avoid}

\begin{figure}[ht]
  \centering
  \includegraphics[width=10cm]{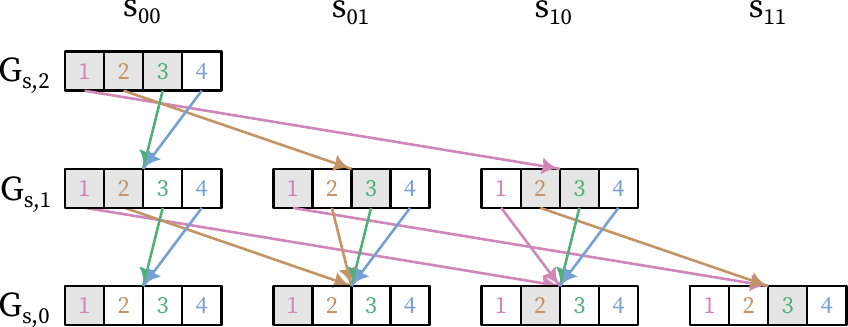}
  \caption{
    In the proof of \Cref{thm:lb-deterministic}, the quantities $F_\sigma$ (defined in Eq. \ref{eq:fsigma-def}) are entirely determined by the values of $\Sigma[\sigma]$ and $r - |\sigma|$. More precisely, we
    have $F_{\sigma} = G_{\Sigma[\sigma],r - |\sigma|}$, where $G_{s,i} := \{\omega_x : \exists \alpha \in [n]^i : \tau(s, \alpha) = x\}$. This diagram shows the values of $G_{s,i}$ for \Cref{alg:trivial} solving the $\mif(4,2)$ problem. The sets $G_{s,i}$ are represented by the dark squares in the array of four cells. The transition function between states is indicated by the colored arrows; for example, green colored arrows (those emitting from squared numbered with a 3) correspond to transitions where the next stream element is a 3, i.e, from state $s$ to state $s' = \tau(s, 3)$. 
  }
  \label{fig:det-diagram}
\end{figure}

\begin{theorem}\label{thm:lb-deterministic}
Every deterministic streaming algorithm for $\mif(n,r)$ requires $\Omega(\sqrt{r} + \frac{r}{1 + \log(n/r)})$ bits of space.
\end{theorem}

\begin{proof}
  Let $\Sigma$ be the set of states of the algorithm, and let $s_\text{init}$ be the initial state. Let $\tau : \Sigma \times [n]^\star \mapsto \Sigma $ be the transition function of the algorithm, where $\tau(s, e_1,\ldots,e_k) = x$ means that if the algorithm is at state $s$, and the next $k$ elements in the stream are $e_1,\ldots,e_k$, then after processing those elements the algorithm will reach state $x$. For each partial stream $\sigma \in [n]^\star$, abbreviate $\tau(s_\text{init}, \sigma)$ as $\Sigma[\sigma]$. For each state $s \in \Sigma$, we
  associate the output $\omega_s \in [n]$ which the algorithm would emit if the state is 
  reached at the end of the stream. (If there is no stream of length $r$ leading to state $s$, we let $\omega_s$ be arbitrary.)
  
  For each partial stream $\sigma \in [n]^\star$, let 
  \begin{align}
    F_\sigma = \{i : \exists  x \in \Sigma, \exists \alpha \in [n]^{r - |\sigma|}: \tau(\Sigma[\sigma], \alpha) = x \land \omega_x = i\} \label{eq:fsigma-def}
  \end{align}
  be the set of possible outputs of the algorithm when $\sigma$ is extended to a stream of length $r$. Because there are only $|\Sigma|$ states, and only $[n]$ possible output values, $|F_s^i| \le m$, where $m = \min(|\Sigma|, n)$.
  
  Let $t,q$ be integers chosen later, so that 
  \begin{align}
    t q \le r - m / 2^q \,.\label{eq:detlb-tq}
  \end{align} We claim that there exists a partial stream $\sigma \in [n]^\star$ satisfying
  $\forall \alpha \in [n]^t : |F_{\sigma.\alpha}| \ge \frac{1}{2} |F_\sigma|$.
  
  Such a state can be found by an iterative process. Let $\tau_0$ be the empty stream $\epsilon$; for $i = 1,2,3,....$ , if there exists $\alpha \in [n]^t$ for which $|F_{\tau_i.\alpha}| \le \frac{1}{2} |F_{\tau_i}|$, let $\tau_{i+1} = \tau_i.\alpha$. Otherwise, stop, and let $\sigma = \tau_{i}$. This process must terminate before $i = q$, because otherwise we would have $|F_{\tau_q}| \le m / 2^q \le r - q t$. Then letting $\gamma \in [n]^{r-q t}$ be a sequence of elements containing every element of $F_{\tau_q}$, we observe that the
  algorithm cannot possibly output a correct answer for the stream $\tau_q.\gamma$. By the definition of $F_{\tau_q}$, we must have $\omega_{\tau_q.\gamma} \in F_{\tau_q}$; but to be a correct missing item finding solution, we need $\omega_{\tau_q.\gamma} \notin \gamma$, hence $\omega_{\tau_q.\gamma} \notin F_{\tau_q}$, a contradiction. Thus, $\sigma = \tau_{i}$ for some $i \le q - 1$. Thus $|\sigma| \le (q - 1) t \le r - t$, which ensures that the terms $\sigma.\alpha$ are streams of length $\le r$ and therefore well defined. Finally, the stopping condition of the process implies $\forall \alpha \in [n]^t : |F_{\sigma.\alpha}| \ge \frac{1}{2} |F_\sigma|$.
  
  We will now construct a deterministic protocol for $\avoid(|F_\sigma|, t, \ceil{\frac{1}{2} |F_\sigma|})$ using $\le \log |\Sigma|$ bits of communication. Alice, on being given a set $A \in \binom{F_\sigma}{t}$, 
  arbitrarily orders it to form a sequence $\alpha$ in $(F_\sigma)^t$; and then sends the state $s' = \tau(\Sigma[\sigma], \alpha)$ to Bob. This can be done using $\log |\Sigma|$ bits of space.
  Bob uses the encoded state to find $F_{\sigma.\alpha}$, by evaluating $\omega_{\tau(s', \beta)}$ for all sequences $\beta \in [n]^{|\sigma| - t}$, and reports the first $\ceil{\frac{1}{2} |F_\sigma|}$ elements of this set as $B$. This protocol works because
  as claimed above, we are guaranteed $|F_{\sigma.\alpha}| \ge |F_\sigma|$; and furthermore, $F_{\sigma.\alpha}$ must be disjoint from $A$; if it is not, then there exists some continuation of $\sigma$ concatenated with $\alpha$ which leads the algorithm to a state $z$ with $\omega_z \in A$, contradicting the correctness of the $MIF$ protocol.
  Finally, applying the communication lower bound from Lemma \ref{lem:avoid-lb}, we find
  \begin{align}
    \log |\Sigma| \ge \frac{1}{\ln 2} t \ceil{\frac{1}{2} |F_\sigma|} / |F_\sigma| \ge t / (2 \ln 2) \label{eq:ldet-lb-rec}
  \end{align}

  We now determine values of $t$ and $q$ satisfying Eq. \ref{eq:detlb-tq}. We can set
  \begin{align*}
    q = \ceil{1 + \log(m/r)} \qquad \text{ and } \qquad t = \floor{ \frac{1}{q}\left( r - \frac{m}{2^q}\right)}
  \end{align*}
  We must have $m \ge r + 1$, as otherwise $|F_\epsilon| \le m \le r$, in which case we could easily make
  the algorithm give an incorrect output by running it on a stream $\gamma \in [n]^r$ which contains all elements
  of $F_\epsilon$. Thus $\log(m/r) \ge 0$, and hence $q \ge 1$, making $t$ well defined. Since $m = \min(|\Sigma|, n)$, we are also guaranteed $\log |\Sigma| \ge \log(r + 1)$. Combining this with Eq. \ref{eq:ldet-lb-rec} gives:
  \begin{align}
    \log |\Sigma| &\ge \max\left(\log(r+1), \frac{1}{2 \ln 2} \floor{ \frac{1}{q}\left( r - \frac{m}{2^q}\right)}\right) \nonumber\\
      &\ge \max\left(1, \frac{1}{2 \ln 2} \floor{ \frac{r}{2 q} }\right) && \hspace{-1cm}\text{since $2^q \ge 2 m /r$ and $r \ge 1$} \nonumber\\ 
      &\ge \frac{1}{1 + 2 \ln 2} \cdot \frac{r}{2 q} && \hspace{-1cm}\text{since $\min(1,(z-1)/y) \ge \frac{z}{1+y}$} \nonumber\\
      &\ge \frac{r}{10 + 5 \log(m/r)} \,. && \hspace{-1cm}\text{since $1 + 2 \ln 2 \le 5/2$} \label{eq:detlb-sigma-ge-mr}
  \end{align}
  As $m = \min(|\Sigma|, n)$, we have $m \le |\Sigma|$, so
  \begin{align*}
    \log |\Sigma| \ge \frac{r/5} {2 + \log |\Sigma| - \log r} \qquad \implies \qquad (\log |\Sigma|)^2 + (2 - \log r) \log |\Sigma| - r/5 \ge 0 \,.
  \end{align*}
  Solving the quadratic inequality gives:
  \begin{align*}
    \log |\Sigma| \ge \sqrt{\frac{r}{5} + \left(1 - \frac{\log(r)}{2}\right)^2} - \left(1 - \frac{\log(r)}{2}\right)
        \ge 
        \begin{cases}
          \sqrt{r/5} & \text{if $r \ge 4$} \\
          0 & \text{otherwise}
        \end{cases}
  \end{align*}
  As $\log |\Sigma| \ge \log(r+1) \ge \sqrt{r/5}$ also holds for $r \le 4$, it follows that $\log |\Sigma| \ge \sqrt{r/5}$ for all values of $r$. Combining this result, Eq. \ref{eq:detlb-sigma-ge-mr}, and the inequality $m \le n$, we conclude:
  \begin{align*}
    \log |\Sigma| \ge \max\left(\sqrt{\frac{r}{5}}, \frac{r}{10 + 5 \log(n/r)}\right) = \Omega\left(\sqrt{r} + \frac{r}{1 + \log(n/r)} \right) \,.
  \end{align*}
\end{proof}

Note: instead of associating "forward" looking sets of outputs $F_s$ with each state $s \in \Sigma$, we could instead use "backward" looking states $B_s$ defined (roughly) as $[n] \setminus \{i : \exists \text{$\sigma$ leading to $s$ with $i \in \sigma$} \}$.

\subsection{Upper bound: the iterated pigeonhole algorithm}

\begin{figure}[ht]
  \centering
  \includegraphics[width=8cm]{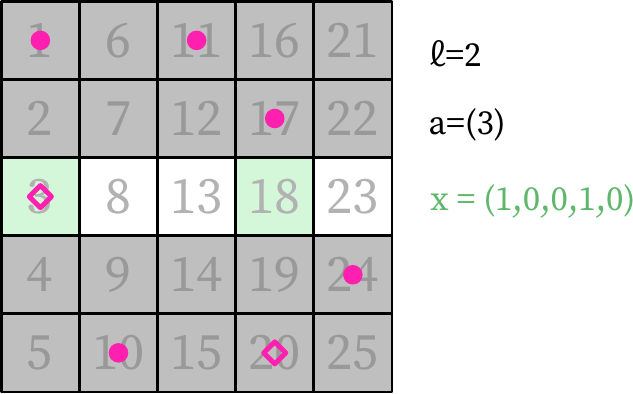}
  \caption{This diagram shows the behavior of \Cref{alg:iterated-pigeonhole}, with $s = 5$ and $t = 2$, on an example input. The pink circles and diamonds mark the elements currently covered by the stream. Cells shaded dark gray are those which are no longer possible outputs due to the current values of $\ell$ and $a$. Cells shaded light green are no longer possible outputs due to the value of the vector $x$. Cells shaded white are possible output values. The algorithm proceeds in $t$ phases; in this example, for the first phase, it maintained a bit vector tracking which of the $s$ rows of the set $[n]$ contained an element from the stream; after the first five elements (1, 10, 11, 17, 24 in some order) arrived, only one row was left available, and the algorithm proceeded to the second phase -- maintaining a bit vector $x$ that records which columns within the chosen row may be unavailable.}
\end{figure}

\begin{theorem}\label{thm:ub-deterministic}
\Cref{alg:iterated-pigeonhole} is a deterministic algorithm that solves $\mif(n,r)$ using $O(\sqrt{r \log r} + \frac{r \log r}{\log n})$ bits of space.
\end{theorem}

\begin{algorithm}[!ht]
  \caption{A deterministic algorithm for $\mif(n,r)$
    \label{alg:iterated-pigeonhole}}
  
  \begin{algorithmic}[1]
  \Statex Let $s,t$ be integers satisfying $s^t \le n$, and $t(s-1) \ge r$.
  \Statex 
  
  \Statex \ul{\textbf{Initialization}}:
    \State $x \gets (0,\ldots,0)$ is a vector in $\{0,1\}^s$.
    \State $(\ell,a) \gets (1,0)$ is an element of $\bigcup_{j \in [t]} \{j\} \times \{0,\ldots,s^j-1\}$
  
  \Statex
  \Statex \ul{\textbf{Update}($e \in [n]$)}:
    \State Let $i \gets \left(\floor{(e-1) / s^{\ell-1}} \bmod s\right) + 1$ \label{line:iter-pigeon-extract}
    \State $x_{i} \gets 1$ \label{line:iter-pigeon-set1}
    \If{$\ell < t$ and there is exactly one $y \in [s]: x_y = 0$}\label{line:iter-pigeon-condition}
      \State $x \gets (0,\ldots,0)$ \label{line:iter-pigeon-reset}
      \State $\ell \gets \ell + 1$ \label{line:iter-pigeon-update-ell}
      \State $a \gets a + (y-1) s^{\ell-1}$ \label{line:iter-pigeon-update-a}
    \EndIf
  
  \Statex
  \Statex \ul{\textbf{Query}}:
    \State Let $i$ be the least value in $[s]$ for which $x_i = 0$
    \State \textbf{output}: $ a + (i-1) s^{\ell-1} + 1$\label{line:iter-pigeon-output}
  
  \end{algorithmic}
    
\end{algorithm}

\begin{proof}
  First, we establish that the variables $(\ell,a)$ of the algorithm stay in their specified bounds. The condition in Line \ref{line:iter-pigeon-condition} ensures that $\ell$ will not be increased beyond $t$. At the time Line \ref{line:iter-pigeon-update-a} is executed, $a < s^{\ell - 1}$; since $y \in [s]$, it follows $a + (y-1)s^{\ell -1} < (1 + (s - 1)) s^{\ell - 1} \le s^\ell$, so the pair $(\ell,a)$ stays in $\bigcup_{j \in [t]} \{j\} \times \{0,\ldots,s^j - 1\}$.
  
  Next, we establish that the algorithm is correct; that the output from Line \ref{line:iter-pigeon-output} does not overlap with current stream $e_1,\ldots,e_k$. For each element $e_j$ in the stream, let $\ell_j$ be the value of $\ell$ at the time the element was added (i.e., at the start of the Update function). For all $h \in [t]$, define $C_h := \{ j \in [t] : \ell_j = h \}$ to indicate the elements for which $\ell_j = h$. Because Line \ref{line:iter-pigeon-condition} only triggers when $x$ has one zero entry, and $x$ is reset to the all zero vector immediately afterwards, and each new element sets at most one entry of $x$ to $1$ (Line \ref{line:iter-pigeon-set1}), we have $|C_h| \ge s - 1$ for all $h$ less than or equal to the current value of $\ell$.
  
  Let $c = a + (i-1) s^{\ell-1}$ be the current output of the algorithm (Line \ref{line:iter-pigeon-output}), minus 1. Note that $c \le s^t - 1 \le n - 1$, so the output is in $[n]$. The value of $c$ can be written in base $s$ as $(c_1,\ldots,c_t)$, so that $c = \sum_{j = 1}^t c_j s^{j -1}$. For $h$ less than the current value of $\ell$, $c_h$ is equal to the value of $y$ at the time the condition of Line \ref{line:iter-pigeon-condition} evaluated to true; in other words, at that time, $x_{c_h} = 0$. Now, for each $j \in C_h$, consider the value $e_j - 1$ written in base $s$ as $(b_1,\ldots,b_t)$. When $e_j$ was added, Line \ref{line:iter-pigeon-extract} set $i$ equal to $b_h$, and so Line \ref{line:iter-pigeon-set1} ensured $x_{b_h} = 1$. Since $x_{c_h} = 0$ held afterwards, when the condition of Line \ref{line:iter-pigeon-condition} evaluated to true, it follows $b_h \ne c_h$. This implies $e_j - 1\ne c$ holds for all $j \in C_h$. A similar argument will establish that for $j \in C_\ell$, we have $e_j - 1\ne c$; since $C_1 \cup \ldots \cup C_\ell$ contains the entire stream so far, it follows the current output of the algorithm does not equal any of the $\{e_j\}_{j =1}^k$, and is thus correct.
  
  
  Finally, we determine values of $s$ and $t$ which for which the algorithm uses little space. The vector $x$ can be stored using $s$ bits; since there are $\sum_{i=0}^{t-1} s^i \le s^t$ possible values of $(\ell,a)$, this algorithm can be implemented using $\le s + t \log s + 1$ bits of space.
  
  Now let 
  \begin{align*}
    q = \min\left(\sqrt{r \log(r+1)}, \log n\right) \qquad \text{and} \qquad t = \floor{\frac{q}{\log(r + 1)}} \qquad \text{and} \qquad s = \ceil{\frac{r}{t}} + 1 \,,
  \end{align*}
  Because $r \ge \log(r+1)$, and $\log n \ge \log(r + 1)$, it follows $t \ge 1$. This implies $s \le r + 1$. Then $t (s - 1) = t \ceil{r / t} \ge r$, and
  \begin{align*}
    s^t \le (r+1)^{\floor{q/\log(r+1)}} \le (r+1)^{q/\log(r+1)} \le 2^q \le n \,,
  \end{align*}
  so the values of $s$ and $t$ satisfy the required conditions $s^t \le n$ and $t(s-1) \ge r$. With these parameters, the space usage of the algorithm is:
  \begin{align*}
    s + t \log s + 1 &\le \ceil{\frac{r}{t}} + 2 +
                      \floor{\frac{q}{\log(r+1)} } \log(  \ceil{\frac{r}{t}} + 1) \\
          &\le \frac{r}{\floor{q / \log(r+1)}} + 3 + \frac{q}{\log(r+1)}\log(r + 1) \\
          &\le \frac{2 r \log(r+1)}{q} + q + 3 \\
          &= \max\left(2 \sqrt{r \log(r+1)}, \frac{2 r \log(r+1)}{\log n} \right) + \min\left(\sqrt{r \log(r+1)}, \log n\right) + 3 \\
          &= O\left(\sqrt{r \log r} + \frac{r \log r}{\log n}\right) \,.
  \end{align*}
\end{proof}

\section{White box model}\label{sec:whitebox}

\begin{figure}[ht]
  \centering
  \includegraphics[width=10cm]{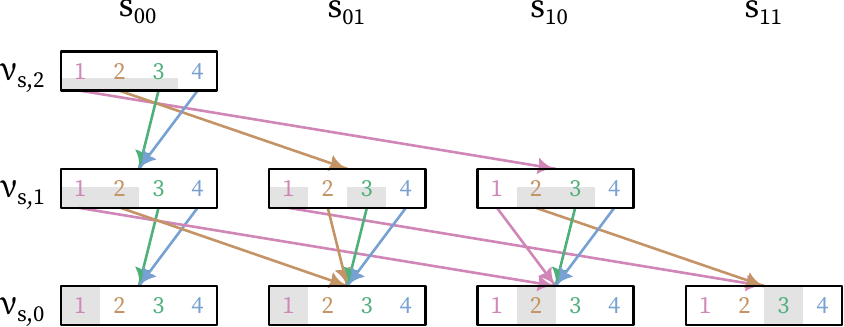}
  \caption{
    In the proof of \Cref{thm:lb-whitebox}, the quantities $\nu_{\sigma,i}$ defined as fixed points of Eq.  \ref{eq:fixedpt-def} are shown for the state diagram of \Cref{alg:trivial} for the problem $\mif(4,2)$.
    The distributions $\nu_{\sigma,i}$ are represented by the
    gray bar charts in each rectangle; for example, the distribution $\nu_{s_{01},1}$ has weight $0.5$ on value $1$ and weight $0.5$ on value $3$. The transition function between states is indicated by the colored arrows; for example, green colored arrows (those emitting from squared numbered with a 3) correspond to transitions where the next stream element is a 3, i.e, from state $s$ to state $s' = \tau(s, 3)$.
  }\label{fig:wb-diagram}
\end{figure}

\begin{theorem}\label{thm:lb-whitebox}
  Every streaming algorithm for $\mif(n,r)$ which has error $\delta \le 1/\left(16 n\right)^{2 \log n + 7} = 1/2^{\Omega(\log n)^2}$ against white-box adversaries requires $\Omega\left(\frac{r}{(\log n)^4}\right)$ bits of space.
\end{theorem}

This proof relies on the following Lemma, whose proof we will defer for later.

\begin{lemma}\label{lem:whitebox-concentration}
  Let $\nu$ be a distribution on $[n]$, and $p = 1 + 1 / \log(n)$. Let $\delta \le \frac{1}{n^3}$. Let $R$ be a random variable with values in $\Omega$. If there is a map $M : [n]^t \times \Omega \rightarrow \triangle[n]$ so that $\EE_{x \sim \nu^t, R} M(x, R) = \nu$, and:
  \begin{align}
    \Pr_{x \sim \nu^t, R}\left[\Vert M(x, R) \Vert_p^p \le D^{p-1} \norm{\nu}_p^p \land \left(\forall j \in [t]: M(x, R)(x_j) \le \delta \right)  \right] \ge 1 - \frac{1}{2^{6} n^2} \label{eq:conc-lem-prob-cond}
  \end{align}
  Then $\log |\range(M)| \ge  \frac{t}{2^{9} D (\log n)^2} - 2 \log(n) - 6$.
\end{lemma}

\begin{proof} Proof of \Cref{thm:lb-whitebox}.

  We can safely assume that $r \ge \log n + 1$, as for any $r = O((\log n)^3)$ the claimed lower bound is trivial.
  
  Let $\ell = \ceil{\log n + 1}$, $t = \floor{\frac{r}{\ell}}$, and let $\hat{r} = t \ell$. We can use a protocol for $\mif(n,r)$ to solve $\mif(n,\hat{r})$ instead, by padding the start of the stream with a fixed sequence of $r - \hat{r}$ arbitrary inputs. Let $\cA$ be this new algorithm.

  Let $\Sigma$ be the set of all states of $\cA$, and let $\tau : \Sigma \times [n] \rightarrow \Sigma$ be the \emph{randomized} transition function between states.
  For each state $s \in \Sigma$, let $\omega_s$ be the distribution over $[n]$ from which
  the final output value is drawn when the final state of the algorithm is $s$.
  (If $s$ can never occur at the end of the stream, we let $\omega_s$ be arbitrary.)
  To each pair $(s,i) \in \Sigma \times \{0,\ldots,\ell\}$, we will associate a distribution $\nu_{s,i}$ over $[n]$. These distributions are recursively defined; if $i = \ell$, we let $\nu_{s,\ell} = \omega_s$, i.e., the output distribution for state $s$. For $i < \ell$, define $f_{s,i} : \triangle[n] \rightarrow \triangle[n]$ as:
  \begin{align}
    f(\phi) = \EE_{x \sim \phi^t} \EE_{s' \sim \tau(s,x)} \nu_{s',i+1} = \sum_{x \in [n]^t} \left(\prod_{i \in [t]} \phi(x_i) \right) \sum_{s' \in \Sigma} \Pr[\tau(s,x) = s'] \nu_{s',i+1} \label{eq:fixedpt-def}
  \end{align}
  Because this function is continuous, and $\triangle[n]$ is homeomorphic to an $(n-1)$-dimensional ball, we can apply Brouwer's fixed point theorem (\Cref{lem:brouwer-fixed-point}) to find a distribution $\nu_{s,i} \in \triangle[n]$ satisfying $\nu_{s,i} = f_{s,i}(\nu_{s,i})$.
  
  With the distributions $\nu_{s,i}$ as defined above, we can define an adversary which, we can prove, will trick $\cA$ into outputting an element that was present in the stream with probability $\ge \frac{1}{(16 n)^{2 \log n + 7}}$. The adversary proceeds in $\ell$ rounds: for each $i \in \{0,\ldots,\ell-1\}$, they identify the current state $s_i$ of the algorithm, sample $\alpha \sim \nu_{s_i,i}$, and send $\alpha$ to $\cA$.
  
  Let $p = 1 + 1 / \log n$; the quantity $\norm{\nu_{s,i}}_p^p$ is a measure of the concentration of the output distribution associated with $s$ and $i$. Assume for sake of contradiction that $\log |\Sigma| \le \frac{t}{2^9 (\log n)^2} - 2 \log n - 6$. Then we shall prove by induction, for all $i \in \{0,\ldots,\ell\}$, the statement $P(i)$ that for all $s \in \Sigma$, if $\norm{\nu_{s,i}}_p^p \ge 2^{i (p-1)}/n^{p-1}$, then the probability that $\cA$ will give an incorrect answer when the remaining $(\ell - i) t$ elements of the stream are provided by the adversary is $\ge 1 / (16 n)^{2(\ell - i) + 3})$. The base case of the induction, at $i = \ell$, holds vacuously, because $\norm{\nu_{s,i}}_p^p \ge 2^{i (p-1)}/n^{p-1} \ge 2^{(\ceil{\log n} + 1) (p-1)} /n^{p-1} \ge 2^{p-1} > 1$ is never true.
  
  Now, for the induction step. Assume $P(i+1)$ holds; we would like to prove $P(i)$ is true. Assume the current state $s$ of the algorithm satisfies $\norm{\nu_{s,i}}_p^p \ge 2^{i (p-1)}/n^{p-1}$. The adversary samples $x \sim \nu_{s,i}^t$ and sends it to the algorithm, which transitions to the state $s' \sim \tau(s, x)$. If it is the case that
  \begin{align}
    \Pr[\text{$\nu$ too concentrated}] := \Pr[\norm{\nu_{s',i+1}}_p^p \ge 2^{(i+1) (p-1)}/n^{p-1}] \ge \frac{1}{2^{7} n^2} \,, \label{eq:wb-conc-path}
  \end{align}
  then, by applying $P(i+1)$, it follows:
  \begin{align*}
     \Pr[\text{$\cA$ fails}] 
     &\ge \Pr[\text{$\cA$ fails} \mid \text{$\nu$ too concentrated}] \Pr\left[\text{$\nu$ too concentrated}\right] \\
     & \ge \frac{1}{(16 n)^{2(\ell - i - 1) + 3}} \cdot \frac{1}{2^{7} n^2} \ge \frac{1}{(16 n)^{2 (\ell - i) + 3}}
  \end{align*}
  It remains to prove $P(i+1)$ assuming Eq. \ref{eq:wb-conc-path} does not hold. If that is the case, let $M : [n]^t \times \Omega \rightarrow \triangle[n]$ be the randomized map in which $M(x, R) = \nu_{s',i+1}$ where $s'$ is randomly chosen according to $ \tau(s,x)$; the random variable $R$ encapsulates the randomness of $\tau$.
  Note that 
  $\EE_{x \sim \nu_{s,i}^t, R} M(x,R) = \nu_{s,i}$, by the definition of $\nu_{s,i}$. Applying Lemma \ref{lem:whitebox-concentration} to $M$, $\nu_{s,i}$, $p$, $D=2$, and $\delta = 1/n^3$
  we observe that since we have assumed that $|\Sigma| \ge \log|\range(M)|$ is smaller than the Lemma guarantees, and $\EE_{x \sim \nu_{s,i}^t} \EE M(x) = \nu_{s,i}$ holds, it must be that Eq. \ref{eq:conc-lem-prob-cond} is incorrect. Thus:
  \begin{align*}
    \Pr_{x \sim \nu^t}\left[ \norm{M(x, R)}_p^p \le 2^{p-1} \norm{\nu_{s,i}}_p^p \land (\forall j \in [t]: M(x, R)(x_j) \le \frac{1}{n^3}) \right] \le 1 - \frac{1}{2^{6} n^2}
  \end{align*}
  and since Eq. \ref{eq:wb-conc-path} does not hold, we have 
  \begin{align*}
    \Pr_{x \sim \nu^t} \left[\norm{\nu_{s',i+1} }_p^p \le 2^{p-1} \norm{\nu_{s,i}}_p^p \right]
      \ge \Pr_{x \sim \nu^t} \left[\norm{\nu_{s',i+1} }_p^p \le 2^{(i+1)(p-1)} / n^{p-1} \right]
      \ge 1 - \frac{1}{2^7 n^2}
  \end{align*}
  which implies:
  \begin{align*}
     \Pr_{x \sim \nu^t, s' \sim \tau(s,x)}\left[\exists j \in [t]:  \nu_{s',i+1} (x_j) \ge \frac{1}{n^3}\right] \ge \frac{1}{2^{7} n^2} \,.
  \end{align*}
  The definition of $\nu_{s',i+1}$ ensures that $\nu_{s',i+1}$ is precisely the distribution of output values if the algorithm and adversary are run for $t (\ell - i - 1)$ steps starting from state $s'$. The probability that the algorithm fails because the final output overlaps with $x$ is then 
  \begin{align*}
    \Pr_{x \sim \nu_{s,i}^t, s' \sim \tau(s,x), y \sim \nu_{s',i+1} }[\exists j \in [t]: x_j = y] &= \EE_{x \sim \nu_{s,i}^t, s' \sim \tau(s,x)} \Pr_{y \sim \nu_{s',i+1}} [\exists j \in [t]: x_j = y] \\
      &= \EE_{x \sim \nu_{s,i}^t, s' \sim \tau(s,x)} \sum_{j \in [t]} \nu_{s',i+1}(x_j) \\
      &\ge \EE_{x \sim \nu_{s,i}^t, s' \sim \tau(s,x)} \max_{j \in [t]} \nu_{s',i+1}(x_j) \\
      &\ge \frac{1}{n^3} \Pr_{x \sim \nu_{s,i}^t, s' \sim \tau(s,x)}\left[ \max_{j \in [t]} \nu_{s',i+1}(x_j) \ge \frac{1}{n^3}\right]\\
      &\ge \frac{1}{n^3}  \cdot \frac{1}{2^{7} n^2} = \frac{1}{2^{7} n^5}
  \end{align*}
  Thus, the failure probability of the algorithm as of $(s,i)$ is $\ge 1 / (2^7 n^5) \ge 1 / (16 n)^5 \ge 1 / (16 n)^{2 (\ell - i) + 3}$; this completes the proof of $P(i)$.
  
  With the proof by induction complete, the statement $P(0)$ implies that for any $s \in \Sigma$, because $\norm{\nu_{s,0}}_p^p \ge \frac{1}{n^{p-1}}$ always holds, the probability that $\cA$ gives an incorrect answer when run against the adversary on a stream of length $t \ell = \hat{r}$ is $\ge 1/(16 n)^{2 \ell + 3} \ge \frac{1}{(16 n)^{2 \log n + 7}}$. This contradicts the given fact that $\cA$'s error is less than this, so the assumption that $\log |\Sigma| \le \frac{t}{2^9 (\log n)^2} - 2 \log n - 6$ must be incorrect; and instead we must have
  \begin{align}
    \log |\Sigma| &\ge \frac{t}{2^9 (\log n)^2} - 2 \log n - 6 \ge \frac{\floor{ r / \ceil{\log n + 1} }}{2^9 (\log n)^2} - 2 \log n - 6 \nonumber\\
      &= \Omega(r / (\log n)^3 - \log n) \,. \label{eq:wb-result}
  \end{align}
  
  To handle the case of small $r$, we note that a white-box algorithm $\cB$ for $\mif(n,r)$ with error $\delta \le 1/\left(16 n\right)^{2 \log n + 7}$
  can be used to solve the $\avoid(n, r, 1)$ communication task. Here, Alice, 
  on being given a set $A \subseteq \binom{[n]}{r}$, runs an instance of $\cB$ on a sequence containing the elements of $A$ in some order; she then sends the state
  of the instance to Bob, who queries the instance for an output, and reports that value. This communication protocol has the same error probability as $\cB$; by \Cref{lem:lb-avoid-low-range}, it requires
  \begin{align*}
    \ge \min\left(\log(r+1), \log \frac{\log{1/\delta}}{\log{e n / r}}\right) \ge \log (\min(r + 1, 2 \log n + 7)) \ge 1
  \end{align*}
  bits of communication; thus $\cB$ requires at least one  bit of state. Since $\max(1, z / a - b) \ge z / (a (1 + b))$, this lets us find a more convenient corollary for Eq. \ref{eq:wb-result}; that $\log |\Sigma| = \Omega(r / (\log n)^4)$.
\end{proof}

We will now prove \Cref{lem:whitebox-concentration}. It relies on the following technical claim about probability distributions; which \emph{roughly} implies that when a distribution is split into a small number of regions on which it is approximately uniform, a specific sum of powers of the weight and density of each region has a lower bound.

\begin{claim}\label{claim:distribution-helper}
  Define $m_\phi(K)$ to be the minimum value of distribution $\phi$ on the set $K$, so $m_\phi(K) := \min_{i \in K} \phi(i)$.

  Let $p > 1$, $\beta \in (0, 1]$, and $n \ge 2$. For any distribution $\nu$ on $[n]$, there exists a collection of disjoint sets
  $\{H_i\}_{i \in J}$ for some $|J| \le \frac{3}{\beta} \log n$ where:
  \begin{align}
    \sum_{i \in J} (m_\nu(H_i))^{p-1} \frac{(\nu(H_i))^p}{\Vert \nu  \Vert_{p}^{p}}
    & \ge \frac{(1-\frac{1}{n})^p}{2^{\beta(2p-1)} |J|^{(p - 1) p / (2p-1)} n^{(p-1)^2 / (2p-1)}} \\
    & \ge \frac{(1-\frac{1}{n})^p}{2^{\beta(2p-1)} |J|^{(p - 1)} n^{(p-1)^2}} \,.\label{eq:distrib-helper-main}
  \end{align}
  Furthermore, we have $\max_{i \in J} m_\nu(H_i) \ge 1 / (n 2^\beta)$, and $\min_{i \in J} m_\nu(H_i) \ge 1/n^2$.
\end{claim}

\begin{proof} (Of \Cref{lem:whitebox-concentration}.)
  In order to avoid awkward expressions like $M(x,R)(i)$, we define $\tmu_x := M(x, R)$. We also use the notation $a^+ := \max(0, a)$. Throughout the proof we shall assume $n \ge 2$, as in the case $n = 1$ it is easy to prove that no such map $M$ exists.
  
  This proof has two main stages. The first establishes that, for a small fraction of vectors $x$ drawn from $\nu^t$, the distribution $\tmu_x$ will probably have significant mass in the same area as $\nu^t$, while not being much more concentrated (according to $\norm{\cdot}_p^p$) than $\tmu_x$, and avoids $x$. The second part will show that such distributions can avoid only small fraction of vectors sampled from $\nu^t$; together, these stages imply the range of $M$ must be large.

  Given a real random variable $W$, with $\EE W \ge y$, and $0 \le W \le \eta y$,
  we have
  \begin{align}
    \Pr[W \ge \alpha y] = 1 - \Pr[W \le \alpha y] \ge 1 - \Pr[(\eta y - W) \ge (\eta - \alpha) y] \ge 1 - \frac{\eta y - y}{(\eta - \alpha) y} \ge \frac{1 - \alpha}{\eta} \,.  \label{eq:reverse-markov}
  \end{align}
  Let $\beta \in (0,1]$ be a parameter chosen later. Apply \Cref{claim:distribution-helper} to $\nu$ with this $\beta$ and the given $p$, producing disjoint sets $\{H_i\}_{i\in J}$.  For any $i \in J$, we have $\EE_{X \sim \nu^t, R} \tmu_X = \nu(H_i)$. Now applying Jensen's inequality to convex functions of the form $f(a) = ((a - b)^+)^p$ gives:
  \begin{align*}  
    \EE_{X \sim \nu^t, R} ((\tmu_X(H_i) - \delta|H_i|)^+)^p &\ge ((\nu(H_i) - \delta|H_i|)^+)^p \qquad \text{which implies} \\
    \EE_{X \sim \nu^t, R}[ \sum_{i \in J} m_\nu(H_i)^{p-1} ((\tmu_X(H_i) - \delta |H_i|)^+)^p ] &\ge \sum_{i \in J} m_\nu(H_i)^{p-1} ((\nu(H_i) - \delta |H_i|)^+)^p \,.
  \end{align*}
  Next, for any $x \in [n]^t$,
  \begin{align*}
    \sum_{i \in J} m_\nu(H_i)^{p-1} ((\tmu_x(H_i) - \delta|H_i|)^+)^p) \le \max_{i \in J} m_\nu(H_i)^{p-1} \le (4 n)^p \sum_{i \in J} m_\nu(H_i)^{p-1} (\nu(H_i) - \delta |H_i|)^p) \,,
  \end{align*}
  because as noted in \Cref{claim:distribution-helper}, for the $i$ maximizing
  $m_\nu(H_i)$, we have $\nu(H_i) \ge |H_i| \frac{1}{n 2^\beta} \ge \frac{|H_i|}{2 n}$, so
  $\nu(H_i) - \delta |H_i| \ge |H_i| (\frac{1}{2 n} - \frac{1}{n^3}) \ge 1/4 n$. Note that $(4 n)^p \le 4^2 (n^{1 + 1/\log n}) = 2^5 n$.
  Applying Eq. \ref{eq:reverse-markov} thus yields:
  \begin{align*}
    \Pr_{ X \sim \nu ^t, R }\left[
      \begin{aligned}  &\sum_{i \in J} m_\nu(H_i)^{p-1} ((\tmu_X(H_i) - \delta|H_i|)^+)^p \quad\ge\\
         &\quad \left(1 - \frac{1}{n}\right) \sum_{i \in J} m_\nu(H_i)^{p-1} ((\nu(H_i) - \delta|H_i|)^+)^p
      \end{aligned}
    \right] \ge \frac{1}{n} \cdot \frac{1}{(4 n)^p} \ge \frac{1}{2^{5} n^2} \,.
  \end{align*}
  Intersecting this event with that of Eq. \ref{eq:conc-lem-prob-cond} implies the probability that all three of the following conditions hold is $\ge \frac{1}{2^{6} n^2}$:
  \begin{align*}
      \text{(a)}:\qquad & \norm{\tmu_X}_p^p \le D^{p-1} \norm{\nu}_p^p \\
      \text{(b)}:\qquad & \forall j \in [t]: \tmu_X(X_j) \le \delta \\
      \text{(c)}:\qquad & \sum_{i \in J} m_\nu(H_i)^{p-1} ((\tmu_X(H_i) - \delta|H_i|)^+)^p \ge (1 - \frac{1}{n}) \sum_{i \in J} m_\nu(H_i)^{p-1} ((\nu(H_i) - \delta|H_i|)^+)^p \,.
  \end{align*}
  By the averaging argument, there must exist a value $R' \in \Omega$ for which, when $R = R'$, the above three conditions hold with at least the same probability. In other words, when replacing $\tmu_X$ with $\mu_X := M(x,R')$, the conditions still holds with probability $\ge 1/{2^6 n^2}$. Now let $G := \{\mu_x : x\in[n]^t \text{ satisfies (a),(c)} \}$ and define $L_\pi := \{ i \in [n] : \pi(i) \le \delta \}$. Therefore,
  \begin{align}
    \frac{1}{2^6 n^2} &\le \Pr_{X \sim \nu^t}[\mu_X \in G \land (\forall j \in [t]: \mu_X(X_j) \le \delta)] \nonumber \\
                      &= \sum_{y \in G} \Pr_{X \sim \nu^t}[\mu_X = \mu_y \land  (\forall j \in [t]: \mu_X(X_j) \le \delta)] \nonumber \\
                      &\le \sum_{y \in G} \Pr_{X \sim \nu^t}[(\forall j \in [t]: \mu_y(X_j) \le \delta)] \nonumber \\
                      &= \sum_{y \in G} \prod_{j \in [t]} \Pr_{X_j \sim \nu}[\mu_y(X_j) \le \delta] = \sum_{y \in G} (\nu(L_{\mu_y}))^t \label{eq:dist-conc-link}
  \end{align}
  We will now prove an upper bound on $\nu(L_{\mu_y})$ for any given $y \in G$. Observe that for any sequence $a_1,\ldots,a_\ell$ of nonnegative real numbers, $\sum_{i = 1}^\ell a_i^p \ge (\sum_{i = 1}^\ell a_i)^p / \ell^{p-1}$; this follows from Hölder's inequality. As the sets $H_i \setminus L_{\mu_y}$ are disjoint,
  \begin{align*}
    \norm{\mu_y}_p^p = \sum_{i \in [n]} \mu_y(i)^p \ge \sum_{i \in J} \frac{\mu_y(H_i \setminus L_{\mu_y})^p}{|H_i \setminus L_{\mu_y}|^{p-1}} \,.
  \end{align*}
  The definition of $L_{\mu_y}$ implies $\mu_y(H_i \setminus L_{\mu_y}) \ge \max(0, \mu_y(H_i) - \delta |H_i|)$. Also, because the minimum value of $\nu$ on $H_i \setminus L_{\mu_y}$ is at least $m_\nu(H_i)$, we have 
  \begin{align*}
    |H_i \setminus L_{\mu_y}| \le \frac{\nu(H_i \setminus L_{\mu_y})}{m_\nu(H_i)} \le \frac{1 - \nu(L_{\mu_y})}{m_\nu(H_i)}
  \end{align*}
  Therefore, 
  \begin{align*}
    \norm{\mu_y}_p^p &\ge \sum_{i \in J} \frac{((\mu_y(H_i) - \delta |H_i|)^+)^p } {(1 - \nu(L_{\mu_y}))^{p-1} / m_\nu(H_i)^{p-1}}\\
      &= \frac{1}{(1 - \nu(L_{\mu_y}))^{p-1}} \sum_{i \in J} m_\nu(H_i)^{p-1} ((\mu_y(H_i) - \delta |H_i|)^+)^p \\
      & \ge  \frac{1 - 1/n}{(1 - \nu(L_{\mu_y}))^{p-1}} \sum_{i \in J} m_\nu(H_i)^{p-1} ((\nu(H_i) - \delta |H_i|)^+)^p \qquad \text{by condition (c)} \\
      & \ge  \frac{1 - 1/n}{(1 - \nu(L_{\mu_y}))^{p-1}} \sum_{i \in J} m_\nu(H_i)^{p-1} \left(\left(1 - \frac{1}{n}\right)\nu(H_i)\right)^p
  \end{align*}
  The last step uses the fact that for all $i \in J$, $\delta |H_i| \le \frac{1}{2 n^3} |H_i| \le \frac{1}{n} |H_i| \min_{j \in \bigcup H_i} \nu(j) \le \frac{1}{n} \nu(H_i)$.
  We now apply condition (a), and divide both sides by $\norm{\nu}_p^p$:
  \begin{align*}
    D^{p-1} &\ge \frac{\norm{\mu_y}_p^p}{\norm{\nu}_p^p} \ge \frac{(1 - 1/n)^{p+1}}{(1 - \nu(L_{\mu_y}))^{p-1}} \sum_{i \in J} m_\nu(H_i)^{p-1} \frac{(\nu(H_i))^p}{\norm{\nu}_p^p} \\
      &\ge \frac{(1 - 1/n)^{p+1}}{(1 - \nu(L_{\mu_y}))^{p-1}} \frac{(1 - 1/n)^{p}}{2^{\beta (2p -1)} (\frac{2}{\beta}\log n)^{(p-1)} n^{(p-1)^2} } \,. \qquad \text{by \Cref{claim:distribution-helper}}
  \end{align*}
  Rearranging this inequality to isolate $\nu(L_{\mu_y})$ reveals:
  \begin{align}
    \nu(L_{\mu_y}) \le 1 - \frac{1}{D} \left( \frac{(1 - 1/n)^{(2 p -1)/(p-1)}}{2^{\beta (2p - 1)/(p-1)} (\frac{2}{\beta}\log n) n^{(p-1)} }\right) 
    \label{eq:ub-on-luy}
  \end{align}
  The right hand side is close to its minimum when $\beta = p - 1 = 1 / \log n$: Thus:
  \begin{align*}
    \nu(L_{\mu_y}) &\le 1 - \frac{1}{D} \left( \frac{(1 - 1/n)^{2 + \log n}}{2^{(2 + \log n) / \log n} (2 (\log n)^2) 2^{\log n / \log n} }\right) \\
      & \le 1 - \frac{1}{D} \left(\frac{(1 - 1/n)^{2 + \log n}}{16 \cdot 2^{2 / \log n}}\right)\frac{1}{(\log n)^2}\\
      & \le 1 - \frac{1}{2^{9} D (\log n)^2}
  \end{align*}
  since $(1 - 1/n)^{2 + \log n} / (16 \cdot 2^{2 / \log n} )$ is increasing in $n$, and when evaluated at $n=2$ gives $2^{-9}$.
  Now we are in a position to simplify Eq. \ref{eq:dist-conc-link}; with this upper bound.
  \begin{align*}
    \frac{1}{2^6 n^2} \le \sum_{y \in G} \left(1 - \frac{1}{2^{9} D (\log n)^2} \right)^t = |G| \left(1 - \frac{1}{2^{9} D (\log n)^2} \right)^t \le |G| \exp\left(- \frac{t}{2^{9} D (\log n)^2}\right) \,.
  \end{align*}
  Since $G$ is a subset of $\range(M)$, we have $\log |\range(M)| \ge \log |G|$; rearranging the above to isolate $|G|$ gives:
  \begin{align*}
    \log |\range(M)| &\ge \log |G| \ge \frac{t}{2^{9} (\ln 2) D (\log n)^2}- \log(2^6 n^2) \\
      & \ge \frac{t}{2^{9} D (\log n)^2} - 2 \log(n) - 6 \,.
  \end{align*}
\end{proof}

Finally, we prove \Cref{claim:distribution-helper}:

\begin{proof}
  For all $i \in \ZZ_{\ge 0}$, let $w_i := 2^{i\beta} / n^2$, and let
  $H_i := \{ j \in [n] : \nu(j) \in [w_i,w_{i+1})$. Define $J := \{ i \in \ZZ_{\ge 0} : H_i \ne \emptyset \}$. We first prove some basic properties of $J$ and the $H_i$.
  \begin{itemize}
    \item If $i \ge \floor{\frac{3}{\beta}\log n}$, then $w_i \ge 2^{\floor{\frac{3}{\beta}\log n} \beta } / n^2 > 2^{\frac{2}{\beta}(\log n) \beta} / n^2 \ge 1$; since $\max_{j\in[n]} \nu(j) \le 1$, it follows such $H_i$ must be empty. Thus $J \subseteq \{0,\ldots,\floor{\frac{3}{\beta}\log n} - 1\}$, and hence $|J| \le \frac{3}{\beta}\log n$.
    \item Because $\min_{j \in [n]} \nu_(j) \ge 1/n$, we are guaranteed that for some $i$ with $w_i \ge 1/(n 2^\beta)$, $H_i \ne \emptyset$. Thus $\max_{i \in J} m_\nu(H_i) \ge 1 / (n 2^\beta)$. Similarly, $\min_{i \in J} m_\nu(H_i) \ge \min_{i \in J} w_i \ge 1/n^2$.
  \end{itemize}
  
  Now, to prove the main part of the result, Eq. \ref{eq:distrib-helper-main}. Let $K = [n] \setminus  \bigcup_{i\in J} H_i = \{j \in [n] : \nu(j) < 1/n^2 \}$. First, we observe that the contribution of the $j \in K$ to $\norm{\nu}_p^p$ is small and can be easily be accounted for:
  \begin{align*}
      n \frac{1}{n^p} &\le \sum_{j \in [n]} \nu(j)^p = \sum_{j \in  \bigcup_{i\in J} H_i} \nu(j)^p + \sum_{j \in K} \nu(j)^p \le \sum_{j \in \bigcup_{i\in J} H_i} \nu(j)^p + n \left(\frac{1}{n^2}\right)^p \\
        &\le  \sum_{j \in\bigcup_{i\in J} H_i} \nu(j)^p + \frac{1}{n^p} \sum_{j \in [n]} \nu(j)^p \,.
  \end{align*}
  This implies $\norm{\nu}_p^p \le (1 - 1/n^p)^{-1}  \sum_{j \in\bigcup_{i\in J} H_i} \nu(j)^p \le (1 - 1/n)^{-1} \sum_{j \in\bigcup_{i\in J} H_i} \nu(j)^p $.
    
  Now, writing $n_i = |H_i|$, we have $n_i w_i \le \nu(H_i) \ge 2^\beta n_i w_i$, and so:
  \begin{align}
    \sum_{i \in J} (m_\nu(H_i))^{p-1} \frac{(\nu(H_i))^p}{\norm{\nu}_p^p} \ge 
      \frac{(1-1/n) \sum_{i \in J} w_i^{p-1} (w_i n_i)^p}{\sum_{i \in J} n_i (2^\beta w_i)^p} = \frac{1 - 1/n}{2^{\beta p}} \frac{\sum_{i \in J} w_i^{2p-1} n_i^p}{\sum_{i \in J} w_i^{p} n_i} \,.\label{eq:conc-core-niwi}
  \end{align}
  We shall later need the following inequality:
  \begin{align}
    \sum_{i \in J} n_i w_i \ge 2^{-\beta} \sum_{j \in \bigcup_{i\in J} H_i} \nu(j) \ge 2^{-\beta} (\sum_{j\in[n]} \nu(j) - \sum_{j\in K} \nu(j)) \ge 2^{-\beta} \left(1 - \frac{1}{n}\right) \,.\label{eq:niwi-gt-apx-1}
  \end{align}
  With this and properties of the $n_i$, we can lower bound Eq. \ref{eq:conc-core-niwi} by using Hölder's inequality several times:
  \begin{align}
    \sum_{i \in J} w_i^p n_i &= \sum_{i \in J} \left(w_i^{2p-1} n_i^{p}\right)^{\frac{p}{2p-1}} \left(n_i^{-(p-1)} \right)^{\frac{p-1}{2p-1}} \nonumber \\ 
      &\le  \left(\sum_{i \in J} w_i^{2p-1} n_i^{p}\right)^{\frac{p}{2p-1}} \left(\sum_{i \in J} n_i^{-(p-1)} \right)^{\frac{p-1}{2p-1}} \nonumber && \text{by Hölder} \\
      &\le \left(\sum_{i \in J} w_i^{2p-1} n_i^{p}\right)^{\frac{p}{2p-1}} |J|^{\frac{p-1}{2p-1}} && \text{since $n_i \ge 1$} \label{eq:lb-wipni}  \\
    \sum_{i \in J} w_i n_i &= \sum_{i \in J} \left(w_i^{2p-1} n_i^{p}\right)^{\frac{1}{2p-1}} \left(n_i^{1/2} \right)^{\frac{2p-2}{2p-1}} \nonumber \\ 
      &\le \left(\sum_{i \in J} w_i^{2p-1} n_i^{p}\right)^{\frac{1}{2p-1}} \left(\sum_{i \in J}  n_i^{1/2} \right)^{\frac{2p-2}{2p-1}} \nonumber && \text{by Hölder} \\
      &\le \left(\sum_{i \in J} w_i^{2p-1} n_i^{p}\right)^{\frac{1}{2p-1}} \left(\left(\sum_{i \in J}  n_i\right)^{1/2} |J|^{1/2} \right)^{\frac{2p-2}{2p-1}} \nonumber && \text{by Cauchy-Schwarz} \\
      &\le \left(\sum_{i \in J} w_i^{2p-1} n_i^{p}\right)^{\frac{1}{2p-1}} n^{\frac{p-1}{2p-1}} |J|^{\frac{p-1}{2p-1}}  && \text{since $\sum_{i \in J} n_i \le n$.} \label{eq:lb-wini}
  \end{align}
  Multiplying Eq. \ref{eq:lb-wipni} by Eq. \ref{eq:lb-wini} raised to the $(p-1)$st power gives:
  \begin{align*}
    \left( \sum_{i \in J} w_i^p n_i \right) \left(\sum_{i \in J} w_i n_i \right)^{p-1} 
      & \le \left(\sum_{i \in J} w_i^{2p-1} n_i^{p}\right) |J|^\frac{p-1}{2p-1} (n^\frac{p-1}{2p-1})^{p-1} (|J|^\frac{p-1}{2p-1})^{p-1} \,,
  \end{align*}
  which implies
  \begin{align*}
    \frac{\sum_{i \in J} w_i^{2p-1} n_i^{p}}{\sum_{i \in J} w_i^p n_i } \ge \frac{\left(\sum_{i \in J} w_i n_i \right)^{p-1}}{|J|^{\frac{(p-1) p}{2p-1}} n^\frac{(p-1)^2}{2p-1}} \ge \frac{(1-\frac{1}{n})^{p-1}}{2^{\beta (p-1)} |J|^{\frac{(p-1) p}{2p-1}} n^\frac{(p-1)^2}{2p-1}} \,.
  \end{align*}
  Substituting this into Eq. \ref{eq:conc-core-niwi} gives:
  \begin{align*}
    \sum_{i \in J} (m_\nu(H_i))^{p-1} \frac{(\nu(H_i))^p}{\norm{\nu}_p^p}
      \ge \frac{(1-\frac{1}{n})^p}{2^{\beta (2 p-1)} |J|^{\frac{(p-1) p}{2p-1}} n^\frac{(p-1)^2}{2p-1}} \,.
  \end{align*}
\end{proof}

\section{Random start and pseudo-deterministic models}\label{sec:randstart-and-pd}

\begin{theorem}\label{thm:lb-randstart}
The space needed for an algorithm in the random-start model to solve $\mif(n,r)$ against adaptive adversaries, with error $\le \delta \le \frac{1}{6}$, satisfies $s \ge S^{PD}_{1/3}(n, \floor{r/(2 s+2)})$.
\end{theorem}

This theorem implies that \emph{if} it is the case that $S^{PD}_{1/3}(\mif(n,r)) = \Omega(r^c / \polylog(n))$ for some constant $c > 0$, then it follows that $S^{RS}_{1/6}(\mif(n,r)) = \Omega(r^{c / (1+c)}/\polylog n )$. Specifically, if $s = S^{RS}_{1/6}(\mif(n,r))$, then \Cref{thm:lb-randstart} would imply $s \ge S^{PD}_{1/3}(n, \floor{r/(2 s+2)}) = \Omega((r/s)^c / \polylog(n))$; multiplying both sides by $s^c$ and raising them to the $1/(c+1)$st power gives $s \ge \Omega(r^{c/(1+c)} / \polylog(n))$.

\begin{proof}
  Let $\Sigma$ be the set of all states of the random-start algorithm $\cA$, and let $\cD$ be the distribution of the initial states of the algorithm. Write $B \sim \cA$ to indicate that $B$ is an instance of $\cA$, i.e., with initial state drawn from distribution $\cD$.
  Let $\ell = 2 \ceil{\log(|\Sigma|)} + 2$, and let $t = \floor{r / \ell}$. For any partial stream $\sigma$ of elements, and instance $B$ of $\cA$, we let $B(\sigma)$ be the sequence of $|\sigma|$ outputs made by $B$ after it processes each element in $\sigma$.

  Consider an adversary $E$ which does the following. Given $\sigma$ the stream it has already passed to the algorithm, and $\omega$ the sequence of outputs that $\cA$ produced in response to $\sigma$, the adversary checks if there exists any $x \in [n]^t$ for which
  \begin{align}
    \forall y \in [n]^t : \Pr_{B \sim \cA}[B(\sigma.x) = \omega.y \mid B(\sigma) = \omega] \le \frac{2}{3} \,. \label{eq:wb-splitter}
  \end{align}
  If so, it sends $x$ to $\cA$, appends $x$ to $\sigma$ and the returned $t$ elements to $\omega$, and repeats the process. If no such $x$ exists, then the adversary identifies the $z \in [n]^t$ which maximizes:
  \begin{align}
    \Pr_{B \sim \cA}[B(\sigma.z)\text{ is incorrect} \mid B(\sigma) = \omega] \,. \label{eq:wb-err-prob}
  \end{align}
  and sends it to the algorithm. (The adversary gives up if either the algorithm manages to give a valid output after $z$, or after it has sent $\ell$ sets of $t$ elements to the algorithm.)
  
  We claim that if $\log |\Sigma| < S^{PD}_{1/3}(\mif(n,t))$, then $E$ makes the algorithm fail with probability $\ge 1/6$. There are two ways that $E$ can be forced to give up: if it tries more than $\ell - 1$ times to find a point where there is no $x \in [n]^t$ satisfying Eq. \ref{eq:wb-splitter}, or if the $z$ it sends fails to produce an error.
  
  Assume that the adversary finds a value of $x$ satisfying Eq. \ref{eq:wb-splitter}, for each of the $\ell$ tries it makes. Let $x_1,\ldots,x_\ell \in [n]^t$ be these values, and let $y_1,\ldots,y_\ell  \in [n]^t$ be the outputs of the algorithm. By applying Eq. \ref{eq:wb-splitter} repeatedly, we have:
  \begin{align*}
    \Pr_{B \sim \cA}[B(x_1.\ldots.x_\ell) = y_1.\ldots.y_\ell] &= \Pr_{B \sim \cA}[B(x_1.\ldots.x_\ell) = y_1.\ldots.y_\ell \mid B(x_1.\ldots.x_{\ell-1}) = y_1.\ldots.y_{\ell-1}] \\
          & \qquad \cdot \Pr_{B \sim \cA}[B(x_1.\ldots.x_{\ell-1}) = y_1.\ldots.y_{\ell-1} \mid B(x_1.\ldots.x_{\ell-2}) = y_1.\ldots.y_{\ell-2}] \\
          & \qquad \cdot \Pr_{B \sim \cA}[B(x_1) = y_1] \\
      & \le (2/3)^\ell \,.
  \end{align*}
  Let $C \subseteq \Sigma$ be the set of initial states of the algorithm for which the adversary finds a sequence satisfying Eq. \ref{eq:wb-splitter}, $\ell$ times. Because the algorithm is deterministic after the initial state is chosen, each $s \in C$ has a corresponding transcript $(\sigma_s,\omega_s) \in [n]^{t\ell} \times [n]^{t\ell}$ that occurs when $E$ is run against an instance of $\cA$ started from $s$. Therefore,
  \begin{align*}
    \Pr_{s \sim \cD}[s \in C] &= \sum_{s \in C} \Pr_{s' \sim cD} [s = s'] \le \sum_{s \in C} \Pr_{B \sim \cA}[B(\sigma_s) = \omega_s] \\
      & \le |\Sigma| \left(\frac{2}{3}\right)^\ell \le 2^{\log |\Sigma|} \left(\frac{2}{3}\right)^{2 \log(|\Sigma|) + 2} \le \left(\frac{2}{3}\right)^2 \le \frac{1}{2}
  \end{align*}
  Thus, the chance that $E$ fails to find a point where no $x$ satisfying Eq. \ref{eq:wb-splitter} exists is $\le \frac{1}{2}$.
  
  To bound the second way in which $E$ can fail, we let $(\sigma,\omega)$ be a partial transcript of the algorithm for which no $x \in [n]^t$ satisfies Eq. \ref{eq:wb-splitter}. Assume the probability that $z$ produces an error is $< 1/3$. Then we have:
  \begin{align}
    \forall z \in [n]^t, \exists y_z \in [n]^t: \qquad &\Pr_{B \sim \cA}[B(\sigma.z) = \omega.y_z \mid B(\sigma) = \omega] \ge \frac{2}{3} \label{eq:wb-uniqueness} \\ 
    \forall z \in [n]^t:\qquad &\Pr_{B \sim \cA}[B(\sigma.z)\text{ is correct} \mid B(\sigma) = \omega] \ge \frac{2}{3} \label{eq:wb-correctness} \,.
  \end{align}
  These conditions together imply that $\omega.y_z$ is a correct $\mif(n,r)$ output sequence for $\sigma.z$. As a result, we can use $\cA$'s behavior after $(\sigma,\omega)$ to construct a pseudo-deterministic algorithm $\Psi$ for $MIF(n,t)$. To initialize $\Psi$, we sample an initial state $B \sim \cA$ conditioned on the event that $B(\sigma) = \omega$, and then send the elements of $\sigma$ to $B$. After this, when $\Psi$ receives an element $e$, we send $e$ to $B$, and report the element $B$ outputs as the output of $\Psi$. By Eqs. \ref{eq:wb-uniqueness} and \ref{eq:wb-correctness}, the sequence of outputs produced by $\Psi$ on any input $x$ in $[n]^t$ will, with probability $\ge 2/3$,  be the (valid) output $y_z$. Thus, $\Psi$ solves $MIF(n,t)$ with $\le 1/3$ error -- which, under the assumption that $\log |\Sigma| < S^{PD}_{1/3}(\mif(n,t))$, is impossible.
  Thus the $z$ chosen by the adversary makes the algorithm err with probability $\ge 1/3$, conditional on it having found  $(\sigma,\omega)$ with no $x \in [n]^t$ satisfying Eq. \ref{eq:wb-splitter}. The probability that the the adversary succeeds is then $\ge \nicefrac{1}{3} \cdot \nicefrac{1}{2} = 1/6$; this contradicts the assumption that $\cA$ has error $\le 1/6$ against any adversary, which implies that we must instead have $\log |\Sigma| \ge S^{PD}_{1/3}(\mif(n,t))$.
\end{proof}

We now present a random-start algorithm whose total space with random bits included improves slightly on \Cref{alg:hidden-list}.

\begin{theorem}\label{thm:ub-randstart}
  \Cref{alg:batch-list} solves $\mif(n,r)$ against adaptive adversaries, with error $\delta$, and can be implemented using $O\left(\left(\sqrt{r} + \frac{r^2}{n}\right)\log n\right)$ bits of space, \emph{including} all random bits used.
\end{theorem}

\begin{figure}[ht]
  \centering
  \includegraphics[width=10cm]{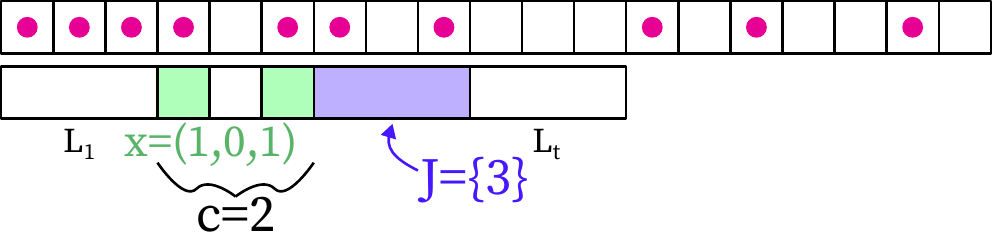}
  \caption{This diagram shows the behavior of \Cref{alg:batch-list} on an example input, if we were to set $w = 3$ and $t = 4$. The top row of squares corresponds to the set $[n]$. In the top row, cells contain a pink dot if the corresponding element has already been seen in the stream. For the bottom row, we have assumed for ease of presentation that $(L_1,L_2,L_3,L_4) = (1,2,3,4)$. The four wide blocks correspond to the sets $\{w (L_j - 1) - 1, \ldots, w L_j\}$ for each $j \in [4]$. The blocks shaded dark blue (here only one) indicate the blocks whose indices are contained in $J$. The vector $x$ tracks which elements in the current block ($L_c$ at $c=2$) were seen in the stream. For indices $> c$, $J$ tracks which blocks contain elements from the stream. If the stream elements in this example had arrived in a different order (say, elements 4 and 6 arriving first), then $J$ might have had the value $\{2,3\}$.
  }
\end{figure}

\begin{algorithm}[!ht]
  \caption{An adversarially robust, "random-start" algorithm for $\mif(n,r)$ with error $\le \delta$
    \label{alg:batch-list}}

  \begin{algorithmic}[1]
  \Statex Assume $r < n / 32$ and $\delta \ge e^{-r / 6}$ -- otherwise, use \Cref{alg:trivial}
  \Statex Let $w = \floor{\min(\sqrt{r \log n}, \frac{n}{32 r}, \frac{r}{6 \ln 1/\delta} )}$, the block size
  \Statex Let $t = \ceil{2 r / w}$
  \Statex 
  
  \Statex \ul{\textbf{Initialization}}:
    \State Let $L = \{L_1,\ldots,L_{t}\}$ be a sequence of $t$ elements from $[\floor{n / w}]$ without repetitions, chosen uniformly at random. 
    \State $c \leftarrow 1$, an integer in the range $\{1,\ldots,t\}$ 
    \State $J \leftarrow \emptyset$, a subset of $[t]$ 
    \State $x \leftarrow (0,\ldots,0)$, a vector in $\{0,1\}^w$ 
    
  \Statex
  \Statex \ul{\textbf{Update}($e \in [n]$)}:
    \State Let $h = \ceil{e / w}$\label{line:block-id}
    \If{$\exists j : L_j = h$ and $j > c$}\label{line:J-increase-cond}
      \State $J \leftarrow J \cup \{j\}$ \label{line:J-increase}
    \EndIf
    \If{$h = L_c$} \label{line:update-x-1}
      \State $x_{e - w (h - 1)} \leftarrow 1$ \label{line:update-x-2}
    \EndIf
    \If{$x = (1,1,\ldots,1)$}
      \State $c \gets c + 1$ \label{line:c-step}
      \While{$c \in J$}
        \State $c \gets c + 1$ \label{line:c-increase}
      \EndWhile
      \State $x \gets (0,0,\ldots,0)$.
    \EndIf
    \If{$c > t$}
      \State \textbf{abort}
    \EndIf
  
  \Statex
  \Statex \ul{\textbf{Query}}:
    \State Let $j$ be the least value in $[w]$ for which $x_j = 0$.
    \State \textbf{output}: $w (L_{c} - 1) + j$
  
  \end{algorithmic}
    
\end{algorithm}

\begin{proof}
  By the $k$th block, we refer to the set $B_k = \{w (L_k - 1) + 1, \ldots, w L_k\}$.

  First, we observe that unless \Cref{alg:batch-list} aborts, it will always output a valid value. At all times, the integer $c$ indicates a value $L_c$ for which the set
  $B_c$ is not entirely contained by the stream.
  The contents of the vector $x$ are updated by Lines \ref{line:update-x-1}
  and \ref{line:update-x-2} to ensure that iff some $e \in B_c$ was given by the stream since the last time $c$ was changed, then the vector entry $x_{e - w(c-1)}$ corresponding to element $e$ has value $1$. On the other hand,
  when the value of $c$ changes, Lines \ref{line:c-step} through \ref{line:c-increase} ensure that for the new value of $c$, no element of $B_c$ was in the stream so far. This works because the variable $J$
  includes (by Lines \ref{line:J-increase-cond} and \ref{line:J-increase}) all
  blocks $B_d$ for $d > c$ which contain a stream element. Thus, after each update,
  $B_c$ always contains at least one element which was not in the stream so far; and since the vector $x$ tracks precisely which elements in $B_c$ were in the stream,
  the query procedure for \Cref{alg:batch-list} always gives a valid result.
  
  Next, we evaluate the probability that \Cref{alg:batch-list} aborts. This only happens if $c > t$. The variable $c$ can increase in two different ways: on Line \ref{line:c-increase}, which can happen at most once every $w$ elements when the vector $x$ fills up; and on Line \ref{line:c-increase}, which occurs at most once for each element in $J$. Thus $c \le 1 + \floor{r/w} + |J|$.
  
  In much the same way that \Cref{thm:ub-advrobust} bounded $|J|$ for \Cref{alg:hidden-list}, we prove here that $|J| \le t - 1 - \floor{r/w}$ with probability $1-\delta$. Without loss of generality, assume that the adversary is deterministic, and picks
  the next element of the stream as a function of the outputs of the algorithm so far.
  Denote the elements of the stream by $e_1,\ldots,e_r$ -- these are random variables
  depending on the algorithm's random choices. Say that $i - 1$ elements have been processed so far, and the algorithm receives the $i$th element. Let $X_i$ be the indicator random variable for the event that the size of $J$ will increase. Matching Line \ref{line:block-id}, let $h_i = \floor{e_i / w}$. Abbreviate $H_{< i} := \{h_1,\ldots,h_{i-1}\}$, $L_{\le c_i} := \{L_1,\ldots,L_{c_i}\}$, $L_{> {c_i}} := \{L_{{c_i}+1},\ldots,L_t\}$; here $c_i$ is the value of the variable $c$ as of Line \ref{line:block-id}. Then by Line \ref{line:J-increase-cond}, $X_i = 1$ iff $h_i \in L_{> c_i} \setminus H_{< i}$. 
  
  Critically, $H_{< i}$ and $h_i$ only depend on what the adversary has seen so far -- algorithm outputs whose computation has only involved $L_1,\ldots,L_{c_i}$ -- and not on the contents of $L_{> c_i}$. Conditioning on $(X_1,\ldots,X_{i-1})$ does constrain
  $L_{> c_i}$, but only in that it fixes the value of $L_{> c_i} \cap H_{< i}$. If we condition on $L_{\le c_i}$ and $(X_1,\ldots,X_{i-1})$ (and hence also on $h_i$ and $H_{< i}$), then the set $L_{> c_i} \setminus H_{< i}$ is a uniform random subset of $[\floor{n/w}] \setminus L_{\le c_i} \setminus H_{< i}$. The probability that $h_i$ is contained in $L_{> c_i} \setminus H_{< i}$ and will be added to $J$ is then:
  \begin{align*}
    \Pr[X_i = 1\mid X_1,\ldots,X_{i-1},L_{\le c_i}] &= 
      \begin{cases}
        0 & h_i \in H_{<i} \cup L_{\le c_i} \\
        \frac{t - c_i - |H_{<i} \cap L_{>c_i}|}{\floor{n/w} - c_i - |H_{<i} \setminus L_{\le c_i}|} & \text{otherwise}
      \end{cases} \\
      &\le \frac{t - |H_{<i} \cap L_{>c_i}|}{\floor{n/w} - |H_{<i} \setminus L_{\le c_i}|} \le \frac{t}{\floor{n/w} - i} \le \frac{2 t}{\floor{n/w}} \,.
  \end{align*}
  In the last step, we used the inequality $i \le r \le 16 r \le \frac{1}{2}\floor{n/w}$.  Taking the (conditional) expectation over $L_{\le c_i}$ yields $\EE[X_i \mid X_1,\ldots,X_{i-1}] \le 2 t / \floor{n/w} \le 4 t w / n$  
  Applying the variant on Azuma's inequality, \Cref{lem:azumanoff}, with $z = \max(1, \frac{3 n}{4 r t w}\ln\frac{1}{\delta})$ gives:
  \begin{align*}
    \Pr\left[ \sum_{i = 1}^{r} X_i \ge r \frac{4 t w}{n} (1 + z) \right] &\le \exp( - z^2 / (2+z) r \frac{4 t w}{n} ) \\
          & \le \exp( - \frac{4 z r t w}{3 n} ) \le \delta \,.
  \end{align*}
  We now observe that:
  \begin{align*}
    \frac{4 r t w (1 + z)}{n} &\le \frac{8 r t w}{n} \max\left(1, \frac{3 n \ln\frac{1}{\delta} }{2 r t w}\right) && \text{defn. of $z$}\\
    & \le \max\left(\frac{8 r t w}{n}, 6 \ln(1/\delta) \right)  \\
    & \le \max\left(\frac{32 r^2}{n}, 6 \ln(1/\delta) \right) &&   \text{since $t \le \frac{4 r}{w}$}   \\
    & \le \max\left(\frac{r}{w}, \frac{r}{w} \right) && \text{since $w \le \frac{n}{32 r}$ and $w \le \frac{r}{6 \ln(1/\delta)}$} \\
    & \le t - \floor{r/w} \,. && \text{since $t \ge \frac{2 r}{w}$}
  \end{align*}
  This implies that at the end of the stream,
  \begin{align*}
    \Pr\left[ |J| \ge  t - \floor{r/w}\right] \le \Pr\left[ |J| = \sum_{i = 1}^{r} X_i \ge  \frac{4 r t w (1 + z)}{n} \right] \le \delta \,,
  \end{align*}
  thereby proving that the algorithm aborts with probability $\le \delta$.

  Finally, we compute the space used by the algorithm. The set $L$ can be stored as a list of integers, using $t \log n$ bits; the counter $c$ with $\log n$ bits; set $J$ with $t$ bits; and vector $x$ with $w$ bits. The total space usage $s$ of the algorithm is then:
  \begin{align}
     s &\le t (1 + \log n) + w + \log n \\
      &\le \frac{10 r \log n}{w} + w && \hspace{-2cm}\text{since $t \le \frac{4 r}{w}$ and $\log n \le \frac{2 r \log n}{w}$} \\
      &\le \frac{20 r \log n}{w} && \hspace{-2cm}\text{since $w \le \sqrt{r \log n}$}\\
      &\le \max\left( 40 \sqrt{r \log n}, \frac{600 r^2 \log n}{n}, 240 \log\frac{1}{\delta}\log n \right) \\
      & = O\left(\sqrt{r \log n} + \frac{r^2}{n} \log n + \log\frac{1}{\delta}\log n \right) \,. \label{eq:randstart-ub}
  \end{align}
  
  \Cref{alg:batch-list} requires $r < n/32$ and $\delta \ge e^{- r /6}$. If either of these conditions do not hold, then it is better to use \Cref{alg:trivial} instead. When $r > n / 32$, we have $\frac{r^2}{n} \log n = \Omega(n \log n)$, and when $\delta \le e^{-r/6}$, we have $\log\frac{1}{\delta}\log n = \Omega(r \log n)$, so using \Cref{alg:trivial} here does not worsen the upper bound from Eq. \ref{eq:randstart-ub}. Consequently, by choosing the better of \Cref{alg:batch-list} and \Cref{alg:trivial}, we can obtain the upper bound from Eq. \ref{eq:randstart-ub} unconditionally.
\end{proof}

It is possible to reduce the space cost of this even further when $r$ is sufficiently smaller than $n$, by replacing the logic used to find a missing element inside a given block. When $r = O(\log n)$, one can obtain an $O\left(r^{1/3} (\log r)^{2/3} \right)$-space algorithm by changing the block size to be $\hat{w} = \tO(n/r^2)$ instead of $w$, and running a nested copy of \Cref{alg:iterated-pigeonhole} configured for $\mif(\hat{w}, r^{2/3} (\log r)^{1/3})$ inside each block instead of tracking precisely which elements in $\{w (L_c - 1) + 1, \ldots, w L_c\}$ have been seen before.

\section{Acknowledgements}

We thank Amit Chakrabarti and Prantar Ghosh for many helpful discussions.

\bibliographystyle{alpha}
\bibliography{refs}

\appendix
\section{Appendix}\label{sec:appendix}

\begin{proof} In this proof of \Cref{lem:azumanoff}, we essentially repeat the proof
  of the Chernoff bound, with slight modifications to account for the dependence of $X_i$ on its predecessors. Here $t$ is a positive real number chosen later.
  \begin{align*}
    \Pr& \left[\sum_{i=1}^n X_i \ge n p (1 + \delta)\right]  \\
      &= \Pr\left[e^{t\sum_{i=1}^n X_i} \ge e^{t n p (1 + \delta)}\right] \\
      &\le \EE[e^{t\sum_{i=1}^n X_i}] / e^{t n p (1 + \delta)} \\
      &= e^{- t n p (1 + \delta)} \EE[e^{t X_1} \EE[e^{t X_2} \cdots \EE[e^{t X_n} \mid X_1 = X_1, \ldots, X_{n-1} = X_{n-1}] \mid X_1 = X_1 ] \\
      &\le e^{- t n p (1 + \delta)} \EE[e^{t X_1} \EE[e^{t X_2} \cdots \EE[e^{t X_{n-1}} (p e^t + (1 - p)) \mid X_1 = X_1, \ldots, X_{n-2} = X_{n-2}] \mid X_1 = X_1 ] \\
      &\le e^{- t n p (1 + \delta)} \EE[e^{t X_1} \EE[e^{t X_2} \cdots \EE[e^{t X_{n-2}} (p e^t + (1 - p))^2 \mid X_1 = X_1, \ldots, X_{n-3} = X_{n-3}] \mid X_1 = X_1 ] \\
      &\le e^{- t n p (1 + \delta)} (p e^t + (1 - p))^n = \left(\frac{p e^t + (1-p) }{e^{t p (1 + \delta)}}\right)^n \\
      & \le \left(\frac{e^{p (e^t - 1)} }{e^{t p (1 + \delta)}}\right)^n = \left(\frac{e^{e^t - 1} }{e^{t (1 + \delta)}}\right)^{n p} \qquad \text{since $1+x \le e^x$} \\
      &= \left(\frac{e^\delta }{(1 + \delta)^{1 + \delta}}\right)^{n p} \qquad \text{picking $t = \ln(1 + \delta)$} \\
      &\le \exp\left(-\frac{\delta^2 n p}{2+\delta}\right) \,. \qquad \text{since $x - (1+x) \ln(1+x) \le -x^2 / (2+x)$}
  \end{align*}
\end{proof}

\end{document}